\documentclass[onecolumn,12pt]{IEEEtran}

\usepackage{cite}
\usepackage{amsmath}
\usepackage{amsfonts}
\usepackage{pifont}
\usepackage{amssymb}
\usepackage{amsthm}
\usepackage{tikz}
\usepackage{graphicx}
    \graphicspath{{../}}
    \DeclareGraphicsExtensions{.pdf}
\usepackage[caption=false,font=footnotesize]{subfig}
\usepackage{multirow}
\usepackage{url}
\usepackage{algorithmic}
\usepackage{enumerate}
\theoremstyle{plain}

\newtheorem{theorem}{Theorem}
\newtheorem{lemma}[theorem]{Lemma}
\newtheorem{proposition}[theorem]{Proposition}
\newtheorem{definition}[theorem]{Definition}
\newtheorem{corollary}[theorem]{Corollary}
\theoremstyle{definition}
\newtheorem{example}{Example}
\newtheorem{construction}{Construction}
\newtheorem*{remark}{Remark}
\DeclareMathOperator*{\Max}{\text{\upshape Max}}
\DeclareMathOperator*{\Min}{\text{\upshape Min}}
\renewcommand{\vec}[1]{\boldsymbol{#1}}

\begin{document}
\title{An Integer Programming Based Bound for Locally Repairable Codes}

\author{\IEEEauthorblockN{Anyu~Wang and
        Zhifang~Zhang}

\IEEEauthorblockA{Key
Laboratory of Mathematics Mechanization, NCMIS\\
Academy of Mathematics and Systems Science, CAS, Beijing, China\\
Email: \{wanganyu, zfz\}@amss.ac.cn}
}
\maketitle
\thispagestyle{empty}

\begin{abstract}
The locally repairable code (LRC) studied in this paper is an $[n,k]$ linear code of which the value at each coordinate can be recovered by a linear combination of at most $r$ other coordinates. The central problem in this work is to determine the largest possible minimum distance for LRCs.
First, an integer programming based upper bound is derived for any LRC. Then by solving the programming problem under certain conditions, an explicit upper bound is obtained for LRCs with parameters $n_1>n_2$, where $n_1 = \left\lceil \frac{n}{r+1} \right\rceil$ and $n_2 = n_1 (r+1) - n$. Finally, an explicit construction for LRCs attaining this upper bound is presented over the finite field $\mathbb{F}_{2^m}$, where $m\geq n_1r$. Based on these results, the largest possible minimum distance for all LRCs with $r \le \sqrt{n}-1$ has been definitely determined, which is of great significance in practical use.

\end{abstract}

\section{Introduction}
In distributed storage systems, redundancy must be introduced to protect data against device failures. The simplest form of redundancy is {\it replication}. But it is extremely inefficient due to its large storage overhead, namely, $c$ copies of the data have to be stored to guarantee $(c-1)$-erasure tolerance. To improve the storage efficiency, {\it erasure codes} are employed in distributed storage systems, such as Windows Azure \cite{Azure2012}, Facebook's Hadoop cluster \cite{XorbasVLDB}, etc, where the original data are divided into $k$ equal-size fragments and then encoded into $n$ fragments $(n>k)$ stored in $n$ different nodes. The fault tolerance property of the erasure code ensures that the system can tolerate up to $d-1$ node failures, where $d$ is the minimum distance of the erasure code. Particularly,
the MDS code is a kind of erasure code that attains the optimal minimum distance with respect to the Singleton bound and thus provides the highest level of fault tolerance for given storage overhead. But the MDS code is still inefficient for distributed storage systems because of the disk I/O complexity it causes in the {\it node repair} issue.
Specifically, when an $[n,k]$ MDS code is employed, repairing a failed node  usually needs the access of $k$ other survival nodes, which entails too much complexity in contrast with the amount of data to be repaired.

To improve this, Gopalan et al. \cite{gopalan2012locality}, Oggier et al. \cite{oggier2011self}, and Papailiopoulos et al. \cite{papailiopoulos2012simple} introduced {\it repair locality} for
erasure codes. The $i$th coordinate of a code has repair locality $r$ if the value at this coordinate can be recovered by accessing at most  $r$ other coordinates. In more detail, a code is said to have {\it information locality} if the locality $r$ is ensured for each coordinate in an information set containing information symbols, e.g., systematic coordinates in a linear systematic code.
Alternatively, a code is said to have {\it all symbol locality} if the locality $r$ is ensured for all coordinates. In this paper we call an $[n,k]$ linear code with all symbol locality $r$ as a {\it locally repairable code} (LRC). When $r \ll k$ it greatly reduces the disk I/O complexity for repair.

Considering the fault tolerance level, the minimum distance is also
an important metric for LRCs. Gopalan et al. \cite{gopalan2012locality} first derived the following upper bound for codes with information locality:
\begin{equation}\label{EqGoplanBnd}
d \le n-k+1 - (\left\lceil \frac{k}{r} \right\rceil -1)
\end{equation}
which is a tight bound by the construction of pyramid codes \cite{Pyramid}.
Although the bound (\ref{EqGoplanBnd}) certainly holds for LRCs, it is not tight in many cases. The results in \cite{gopalan2012locality} pointed out that when $(r+1)\nmid n$ and $r \mid k$ the bound (\ref{EqGoplanBnd}) cannot be attained for codes with all symbol locality, and for those attaining this bound only the existence result was given for the case $(r+1)\mid n$ and the finite field needs to be large enough.
Later, in paper \cite{papailiopoulos2012locally} and \cite{forbes2013locality}, the bound (\ref{EqGoplanBnd}) was generalized to vector codes and nonlinear codes. The impact of field size on the minimum distance of LRCs was considered in \cite{cadambe2013upper}.
The result provides an improved upper bound, but relies on a parameter related to another open problem in coding theory.
In order to deal with multiple erasures in local repair, Prakash et al \cite{prakash2012optimal} proposed the locality $(r,\delta)$ associating the coordinate with an inner-error-correcting code with length less than $r+\delta-1$ and minimum distance at least $\delta$. It is evident that the locality $(r,\delta)$ degenerates into the locality $r$ when $\delta=2$. An upper bound was derived in \cite{prakash2012optimal} for codes with information locality $(r,\delta)$ which coincides with the bound (\ref{EqGoplanBnd}) at $\delta=2$, and an explicit code attaining this bound was given for a specific value of the length $n=\lceil\frac{k}{r}\rceil(r+1)$.

For simplicity, the LRC that achieves the bound (\ref{EqGoplanBnd}) with equality is usually called an optimal LRC.
The first explicit optimal LRCs for the case $(r+1)\mid n$ were constructed in \cite{tamo2013optimal} and \cite{silberstein2013optimal} by using Reed-Solomon codes and Gabidulin codes respectively. Both constructions were built over a  finite field of size exponential in the code length $n$. Moreover, it was proved in \cite{silberstein2013optimal} that the construction also induces an optimal LRC when $n\mod(r+1)~>~k\mod r~>~0$. Then in \cite{FamilyTamo14} for the same case $(r+1)\mid n$ the authors constructed an optimal code over a finite field of size comparable to $n$ by using specially designed polynomials. This construction can be extended to the case $(r+1)\nmid n$ with the minimum distance $d\geq n-k-\lceil\frac{k}{r}\rceil+1$ which is at most one less than the upper bound defined in (\ref{EqGoplanBnd}).

Recently, Song et al. \cite{song2014optimal} obtained more results about tightness of the bound (\ref{EqGoplanBnd}). Specifically, they derived a new case where there are no optimal LRCs and two new cases where there exist optimal LRCs over sufficiently large fields, leaving only two cases in which tightness of the bound (\ref{EqGoplanBnd}) is unknown.
Another recent improvement was in \cite{prakash2014codes} where Prakash et al. showed a new upper bound on the minimum distance for LRCs. This bound relies on a sequence of recursively defined parameters and is tighter than the bound (\ref{EqGoplanBnd}). But no general constructions attaining this new bound was presented.

There are lots of other work devoted to the locality in the handling of multiple node failures, such as \cite{wang2013repair,tamo2014bounds,rawat2014locality,FamilyTamo14} considering LRCs which permit parallel access of ``hot data",  the papers \cite{wang2013repair,pamiesjuarez2013locally} studying LRCs with general local repair groups, and the work \cite{prakash2014codes} which proposed sequential local repair. In a word, more and more research work have concerned about codes with the local repair property, especially those codes attaining the largest possible minimum distance.

\subsection{Our Contribution}
Since the bound (\ref{EqGoplanBnd}) is not tight for LRCs in many cases, the central problem in this work is determining the largest possible minimum distance of an $[n,k]$ LRC.

Our first result is an integer programming based upper bound,
\begin{equation*}
d \le n-k+1 - \eta,
\end{equation*}
where $\eta = \max\{x : \Psi(x) - x < k\}$ and the function $\Psi(x)$ relies on an integer programming problem defined below
\begin{equation*}
\Psi(x) = \Max_{\substack{ s,t_1,\dots,t_s \\ a_1,\dots,a_s}}  \Min_{\;\;l,h_1,\dots,h_l} ( xr+1 - \sum_{i=1}^{l-1}(a_{h_i} - t_{h_i})),\;\; \forall 1 \le x \le \left\lceil \frac{n}{r+1} \right\rceil,
\end{equation*}
where the `Max' is subject to
\begin{equation*}
\begin{cases}
t_1 +\dots + t_s = n_1; \\
a_1 + \dots + a_s = n_2; \\
a_i \ge t_i -1, \forall i \in[s];\\
s\geq1; t_i\geq1, \forall i \in[s],
\end{cases}
\end{equation*}
and the `Min' is subject to
\begin{equation*}
t_{h_1}+\dots+t_{h_{l-1}}<x\leq t_{h_1}+\dots+t_{h_{l}}.
\end{equation*}

By solving the integer programming problem when $n_1>n_2$, we get the second result of this paper: an explicit upper bound on the minimum distance (Theorem \ref{ThmBndN1N2}), where $n_1 = \left\lceil \frac{n}{r+1} \right\rceil$ and $n_2 = n_1 (r+1) - n$. This upper bound stands for all possible values of $k$ while most previous results (e.g., \cite{silberstein2013optimal,song2014optimal}) that depend on the value of $k$ in addition to the parameters $n$ and $r$, which means our bound sometimes covers wider parameter region.
Additionally, in Section \ref{SecExpBnd}-B we show by comparisons that this explicit bound can give sharper description of the largest possible minimum distance than previous results (i.e. the results in \cite{gopalan2012locality, prakash2014codes, song2014optimal}) in many cases.

The third result concerns the construction of LRCs.
Specifically, when $n_1>n_2$, we give an explicit construction  (Construction \ref{Cnst}) of the $[n,k]$ LRC  attaining the bound in Theorem \ref{ThmBndN1N2} over the finite field $\mathbb{F}_{2^m}$, where $m\geq n_1r$.
Therefore,  we have definitely determined the largest possible minimum distance for all $[n,k]$ LRCs under the condition $n_1>n_2$.
Since the condition $n \ge (r+1)^2$ implies $n_1 > n_2$,
we have completely obtained the largest possible minimum distance for LRCs with $r \le \sqrt{n}-1$, which is of great significance in practical use.

\subsection{Related Work}
In \cite{wang2014repair}, the authors developed the framework of regenerating sets which  determines the upper bound on the minimum distance for any LRC by computing a function related to the structure of local repair groups. The upper bound derived in this work can be viewed as an optimization based on this framework. A brief introduction of the framework and the motivation for optimization can be found in  Section \ref{SecRegSet}.

\subsection{Organization}
Section \ref{SecRegSet}  introduces the framework of regenerating sets and shows the motivation of optimization.
Section \ref{SecBnds} derives an integer programming based upper bound on the minimum distance for LRCs.
Then Section \ref{SecExpBnd} solves the integer programming problem for $n_1 > n_2$, and obtains an explicit upper bound.
Section \ref{SecCnst} presents an explicit construction attaining this bound. Finally, Section \ref{SecConclusion} concludes the paper.

\section{Regenerating Sets and locally repairable codes}\label{SecRegSet}
Let $\mathcal{C}$ be an $[n,k,d]_q$ linear code with generator matrix $G = (\vec{g}_1, \dots, \vec{g}_n)$, where $\vec{g}_i \in \mathbb{F}_q^{k}$ for $1\leq i\leq n$.
Then the regenerating set introduced in \cite{wang2014repair} can be defined as follows.
\begin{definition}\label{DefRegSet}
For an $[n,k,d]_q$ linear code $\mathcal{C}$, a regenerating set of the $i$th coordinate, $1\leq i\leq n$, is a subset $R\subseteq [n]$ such that $i\in R$ and $\vec{g}_i$ is an $\mathbb{F}_q$-linear combination of $\{\vec{g}_j\}_{j \in R\backslash \{i\}}$, where $[n]$ denotes the set of integers $\{1,2,\dots,n\}$.
\end{definition}

The collection of all regenerating sets of the $i$th coordinate is denoted by $\mathcal{R}_i$. Furthermore, a sequence of regenerating sets $R_1, R_2,\dots,R_m$, where $R_i\in\mathcal{R}_{l_i}$ and $l_i\in[n]$ for $1\leq i\leq m$, is said to have a {\it nontrivial union} if $l_j\notin\cup_{i=1}^{j-1}R_i$ for $1\leq j\leq m$.

For a linear code $\mathcal{C}$, define the function
\begin{equation}\label{EqPhiX}
\Phi(x)=\min\{|\cup_{i=1}^xR_i|: R_i\in\mathcal{R}_{l_i} \mbox{~and~} R_1,\dots,R_x \mbox{ have a nontrivial union}\}.
\end{equation}
In particular, it is assumed $\Phi(0)=0$.
Then it was proved that the minimum distance is closely related to the function $\Phi(x)$.

\begin{theorem}[\upshape \cite{wang2014repair}]\label{ThmRegDis}
For any $[n,k,d]$ linear code,
$d\leq n-k+1-\rho$,
where $\rho=\max\{x : \Phi(x)-x<k\}.$
\end{theorem}

\begin{remark}
An explicit bound from Theorem \ref{ThmRegDis} depends on computation of the function $\Phi(x)$ which is determined by the specific generator matrix.
Sometimes, partial information of the generator matrix may help get a precise estimate of $\Phi(x)$ which in turn gives a tight bound for the minimum distance.
An instance where Theorem \ref{ThmRegDis} derives a tight bound is the square code proposed in \cite{wang2014repair}.
In this paper, we aim to tighten the minimum distance bound for LRCs by estimating $\Phi(x)$ and then optimizing the value.
The following two subsections explain our motivations through examples.
\end{remark}

\subsection{Estimate of $\Phi(x)$}
First, we need to redefine the locality $r$ by using the concept of regenerating sets.
\begin{definition}\label{DefLocR}
For $1\leq i\leq n$, the $i$th coordinate of an $[n,k]$ code $\mathcal{C}$ has locality $r$ if there exists a regenerating set $R \in \mathcal{R}_i$ with $|R| \le r+1$.
\end{definition}
We refer to an $[n,k]$ linear code of which each coordinate has locality $r$ as a locally repairable code (LRC).
Because $r =1$ implies repetition and for $r \ge k$ MDS code possess the optimal distance, we assume $1< r<k$ throughout the paper.
Moreover, because of the upper bound on the information rate of LRCs \cite{FamilyTamo14}, we assume that $\frac{k}{n} \le \frac{r}{r+1}$ for any $[n,k]$ LRC.

In \cite{wang2014repair} the authors estimated the function $\Phi(x)$ for different kinds of locality and reproved the minimum distance bounds that had been given in previous literatures. For example, it proved $\Phi(x)\leq (r+1)x$ for LRCs which induces the bound (\ref{EqGoplanBnd}); $\Phi(x)\leq r\left\lceil\frac{x}{\delta-1}\right\rceil+x$ for codes with locality $(r,\delta)$ and derived the upper bound given in \cite{prakash2012optimal}; etc.

In this paper we focus on LRCs.
The following example shows that when $(r+1)\nmid n$ one can estimate $\Phi(x)$ better than $\Phi(x)\leq (r+1)x$ and thus can derive a tighter bound.

\begin{example}\label{ExBnd}
Let $\mathcal{C}$ be an $[n,k,d]$ LRC with $(r+1) \nmid n$.
We claim that $\Phi(x) \le x(r+1)-1 $ for $x\geq 2$.

First, the following algorithm generates a sequence of regenerating sets  $R_1,\dots,R_{l}$ that has a nontrivial union and $\cup_{i=1}^l R_l=[n]$.

\vspace*{16pt}
\begin{algorithmic}[1]
    \STATE Set $i=1$
    \WHILE{$\cup_{j=1}^{i-1}R_j \subsetneqq [n]$}
        \STATE Pick $i_0 \in [n] - \cup_{j=1}^{i-1}R_j$
        \STATE Choose $R_i \in \mathcal{R}_{i_0}$ such that $|R_i|=r+1$
        \STATE Set $i=i+1$
    \ENDWHILE
\end{algorithmic}
\vspace*{16pt}

Because $(r+1) \nmid n$ and $|R_i|=r+1$ for  $1 \le i \le l$, there exist $i_1,i_2\in[l]$ such that $R_{i_1} \cap R_{i_2} \neq \emptyset$.
By the definition of $\Phi(x)$,  $\Phi(x)\le \min\{|\cup_{i\in I}R_i|:I\subset[l], |I|=x\}$.
Therefore,
\begin{equation*}
\Phi(x) \le \begin{cases} r+1, \text{ if } x=1, \\ x(r+1) -1, \text{ if } x \ge 2.\end{cases}
\end{equation*}
It follows that $ \rho \ge \left\lceil \frac{k+1}{r} \right\rceil -1$, and thus
\begin{equation}\label{EqBndExDist}
d \le n-k+1 - (\left\lceil \frac{k+1}{r} \right\rceil -1).
\end{equation}

Obviously, the bound (\ref{EqBndExDist}) is tighter than the bound (\ref{EqGoplanBnd}) for the case $(r+1)\nmid n$.
Particularly, the difference occurs when $r\mid k$ which also explains a known fact (see \cite{gopalan2012locality,song2014optimal}) that the bound (\ref{EqGoplanBnd}) is unachievable when $(r+1)\nmid n$ and $r\mid k$.
\end{example}
Later in Section III we will give a shaper estimate of $\Phi(x)$ and derive a tighter bound for LRCs.

\subsection{Optimization of $\Phi(x)$}
From Theorem \ref{ThmRegDis} we observe that for a given LRC, its minimum distance $d$ is upper bounded by $n-k+1-\rho$, where $\rho$ depends on the function $\Phi(x)$ which is determined by the code itself.
Therefore, to upper bound $d$ for all LRCs with parameters $n,k,r$, one needs to find the code which gives the minimum $\rho$ or the maximum $\Phi(x)$.
Actually, we find the structure of regenerating sets plays an important role in determining  the function $\Phi(x)$ which in turn influence the minimum distance.

\begin{example}\label{ExConstruct}
Consider LRCs with parameters $n=10,k=5$ and $r=3$.
We construct two such LRCs which have different structure of regenerating sets.

The first code $\mathcal{C}_1$ is constructed by using rank-metric codes \cite{silberstein2013optimal}.
Specifically, let
$$\{\alpha_1, \alpha_2, \alpha_3, \alpha_5, \alpha_6, \alpha_7, \alpha_9 \} \subseteq \mathbb{F}_{2^7}$$
be a basis of $\mathbb{F}_{2^7}$ over $\mathbb{F}_{2}$ and let
\begin{equation*}
\begin{cases}
\alpha_4 = \alpha_1 + \alpha_2 +\alpha_3 \\
\alpha_8 = \alpha_5 + \alpha_6 + \alpha_7 \\
\alpha_{10} = \alpha_9.
\end{cases}
\end{equation*}
The generator matrix of $\mathcal{C}_1$ is
$G_1 = (\vec{g}_1,\vec{g}_2,\dots,\vec{g}_{10})$,
where $\vec{g}_i = (\alpha_i,\alpha_i^2,\alpha_i^4,\alpha_i^8,\alpha_i^{16})^\tau$ for $1 \le i \le 10$.

It is easy to verify that $\mathcal{C}_1$ is an LRC over $\mathbb{F}_{2^7}$ and a sequence of its regenerating sets is
\begin{equation}\label{eq3}\{1,2,3,4\},\{5,6,7,8\},\{9,10\}.\end{equation}
Therefore, $\Phi(x) \le 4x-2$ for $1 \le x \le 3$ and  $\rho \ge 2$.
By Theorem \ref{ThmRegDis} we have  $d \le n-k+1 - \rho \le 4$.
On the other hand, since any $7$ columns of $G_1$ has full rank, it implies $d\geq n-6=4$.
As a result, $\mathcal{C}_1$ has minimum distance $d=4$.

The second code $\mathcal{C}_2$ is an $[n=10,k=5]$ linear code over $\mathbb{F}_{13}$ with generator matrix
\begin{equation*}
G = \begin{pmatrix}
1 & 0 & 0 & 1 & 0 & 0 & 1 & 5 & 5 & 11 \\
0 & 1 & 0 & 1 & 0 & 0 & 0 & 3 & 7 & 10 \\
0 & 0 & 1 & 1 & 0 & 0 & 0 & 10 & 10 & 7 \\
0 & 0 & 0 & 0 & 1 & 0 & 1 & 6 & 3 & 9 \\
0 & 0 & 0 & 0 & 0 & 1 & 1 & 10 & 9 & 6
\end{pmatrix}.
\end{equation*}
Observe that $\mathcal{C}_2$ has locality $r=3$ and a sequence of its regenerating sets is
\begin{equation}\label{eq4}\{1,2,3,4\},\{1,5,6,7\},\{1,8,9,10\}.\end{equation}
Furthermore, it can be verified that $\Phi(1)=4,\Phi(2)=7$ and $\Phi(3)=10$. Then $\rho=1$ and $d\leq n-k+1 - \rho \le 5$ from Theorem \ref{ThmRegDis}.
On the other hand, one can verify  that $\mathcal{C}_2$ has minimum distance $d = 5$.
\end{example}

From (\ref{eq3}) and (\ref{eq4}) we can see that $\mathcal{C}_1$ and $\mathcal{C}_2$ have different structure of regenerating sets. The former has pairwise disjoint regenerating sets while the latter has overlapped regenerating sets. This difference results in that
the $\Phi(x)$ of $\mathcal{C}_1$ is no more than that of $\mathcal{C}_2$, therefore the latter code has a higher upper bound from Theorem \ref{ThmRegDis}.

\section{Upper Bounds on The Minimum Distance}\label{SecBnds}
Denote $n_1 = \left\lceil \frac{n}{r+1} \right\rceil$ and $n_2 = \left\lceil \frac{n}{r+1} \right\rceil (r+1) - n$. It follows that $n=n_1(r+1)-n_2$ and $0\leq n_2<r+1$. The integer programming based upper bound is derived in three steps as described in the following three subsections respectively.

\subsection{From $\Phi(x)$ to a Set Cover Problem}\label{SubSecCov}
First, for any $[n,k]$ LRC, we convert the problem of estimating the $\Phi(x)$ to a set cover problem (Lemma \ref{LemCov1}, Lemma \ref{LemCov2}).
To begin with, we introduce the concept of an $(r+1)$-cover.
\begin{definition}\label{DefCov}
Let $\mathcal{S}=\{S_1,\dots,S_t\}$ be a collection of subsets of $[n]$. We call $\mathcal{S}$ an $(r+1)$-cover over $[n]$ if the following conditions are satisfied:
\begin{itemize}
  \item[(1)] $|S_i|=r+1$ for $1\leq i\leq t$;
  \item[(2)] $\cup_{i\in[t]}S_i=[n]$ and $\cup_{i\in[t]\setminus\{j\}}S_i\neq[n]$ for any $j\in[t]$.
\end{itemize}
\end{definition}
In the remainder of this paper we usually omit the phrase `over $[n]$' for an $(r+1)$-cover when it is evident from the context.

\begin{lemma}\label{LemCov1}
For a given $[n,k]$ locally repairable code $\mathcal{C}$, it induces an $(r+1)$-cover  $\mathcal{S}=\{S_1,\dots,S_{t}\}$, $t\geq n_1$, satisfying $$\Phi(x)\leq \Min_{\substack{J\subseteq [t] \\ |J|=x}} |\cup_{i\in J}S_i|$$ for $1\leq x\leq n_1$, where $\Phi(x)$ is defined as in (\ref{EqPhiX}).
\end{lemma}
\begin{proof}
By using the algorithm in Example \ref{ExBnd}, we get a sequence of regenerating sets $R_1,\dots,R_l$ which has a nontrivial union.
Then by deleting some $R_i$'s which lie in the union of the remainders, we can finally get an $(r+1)$-cover $\{R_{i_1},\dots,R_{i_t}\}$ as required by the lemma.
\end{proof}

\begin{lemma}\label{LemCov2}
For any $(r+1)$-cover  $\mathcal{S}=\{S_1,\dots,S_{t}\}$, $t> n_1$, there exists an $(r+1)$-cover consisting of $n_1$ subsets, denoted as $\mathcal{T}=\{T_1,\dots,T_{n_1}\}$, which satisfies for $1\leq x\leq n_1$, $$\Min_{\substack{J\subseteq [t] \\ |J|=x}} |\cup_{i\in J}S_i|\leq \Min_{\substack{I\subseteq [n_1] \\ |I|=x}} |\cup_{i\in I}T_i|\;.$$
\end{lemma}
\begin{proof}
Since $t>n_1$, set $T_i=S_i$ initially for $1\leq i\leq n_1$. Due to the condition (2) in Definition \ref{DefCov}, it obviously has $\cup_{i=1}^{n_1}T_i\subsetneqq[n]$. Then we recursively invoke the following Step 1 to  Step 3 on the collection $\mathcal{T}=\{T_1,\dots,T_{n_1}\}$ expanding $\cup_{i=1}^{n_1}T_i$ by one element at each invocation until finally $\cup_{i=1}^{n_1}T_i=[n]$.
\begin{itemize}
  \item[]{\bf Step 1.~} Pick $T_j\in\mathcal{T}$ such that $T_j\cap(\cup_{T\in\mathcal{T}\setminus\{T_j\}}T)\neq\emptyset$.
  \item[]{\bf Step 2.~} Choose $a\in T_j\cap(\cup_{T\in\mathcal{T}\setminus\{T_j\}}T)$ and $b\in [n]-\cup_{i=1}^{n_1}T_i$.
  \item[]{\bf Step 3.~} $T_j\leftarrow(T_j-\{a\})\cup\{b\}$.
\end{itemize}

Note that the subset $T_j$ exists in Step 1 because $\sum_{i=1}^{n_1}|T_i|=n_1(r+1)\geq n>|\cup_{i=1}^{n_1}T_i|$. After the three steps, only one element in $T_j$ is replaced by an outside element and all other subsets remain unchanged. Therefore, $\cup_{i=1}^{n_1}T_i$ is expanded by one element. Furthermore, the union size of any $x$ subsets, $1\leq x\leq n_1$, is unchanged or increased by $1$. Therefore, for $1\leq x\leq n_1$,
$$
\Min_{\substack{ J \subseteq [t] \\ |J| = x}} |\cup_{i \in J}S_i| \le \Min_{\substack{I \subseteq [n_1] \\ |I| = x}} |\cup_{i \in I}S_i| \le \Min_{\substack{I \subseteq [n_1] \\ |I| = x}} |\cup_{i \in I}T_i|.
$$

Moreover, the condition $\cup_{i\in[t]\setminus\{j\}}S_i\neq[n]$ for any $j\in[t]$ implies that $S_j\nsubseteq \cup_{i\in [t]\setminus\{j\}}S_i$ for any $j\in[t]$. It is easy to verify that the property $T_j\nsubseteq \cup_{i\in [n_1]\setminus\{j\}}T_i$ for any $j\in[n_1]$ still holds after an invocation of Step 1 to Step 3. Thus we finally get an $(r+1)$-cover $\mathcal{T}$ as the lemma requires.

\end{proof}
By Lemma \ref{LemCov1} and Lemma \ref{LemCov2}, we have transformed the problem of deriving an upper bound for $\Phi(x)$ into the problem of estimating the set union size in an $(r+1)$-cover consisting of $n_1$ subsets. In the sequel, a further investigation into the $(r+1)$-cover helps to finally derive an upper bound of $\Phi(x)$.

\subsection{From the Set Cover to an Integer Programming Problem}\label{SubSecConneted}
Then we transform the set cover problem into an integer programming problem (Lemma \ref{LemIntProb}).
The following definition comes from the concept of connectivity in graph theory.
\begin{definition}\label{DefCon}
Let $\mathcal{S}=\{S_1,\dots,S_t\}$ be a collection of nonempty subsets of $[n]$. We say $\mathcal{S}$ is connected if for any nonempty subset $I\subsetneqq[t]$, it has $(\cup_{i\in I}S_i)\cap(\cup_{j\in[t]\setminus I}S_j)\neq\emptyset$. Particularly, a collection containing only one subset, i.e. $t=1$, is also called connected.
\end{definition}

\begin{remark}
In fact, a collection $\mathcal{S}$ defines a graph $G(V,E)$,  where each vertex $v_i\in V$ corresponds to a subset $S_i\in\mathcal{S}$ and there is an edge $(v_i,v_j)\in E$ if and only if $S_i\cap S_j\neq\emptyset$. Thus a connected collection in
Definition \ref{DefCon} actually corresponds to a connected graph.
\end{remark}

\begin{proposition}\label{PropCon}
For a connected collection of subsets $\mathcal{S}=\{S_1,\dots,S_t\}$, there exists a permutation of $[t]$, say $\{i_1,\dots,i_t\}$, such that
\begin{equation}\label{EqConOrder}
S_{i_j}\cap(\cup_{h=1}^{j-1} S_{i_h})\neq\emptyset, \;2 \le j \le t.
\end{equation}
\end{proposition}
\begin{proof}
In fact, $i_1 ,\dots,i_t$ can be determined by the following algorithm.

\vspace*{16pt}
\begin{algorithmic}[1]
    \STATE Pick $i_1 \in [t]$
    \FOR {$h=2$ \TO $l$}
        \STATE Pick $i_h \in [t] -\{i_1, i_2, \dots,i_{h-1}\}$ such that \hfill \\ \hspace*{6pt} $S_{i_h} \cap (S_{i_1} \cup \dots \cup S_{i_{h-1}}) \neq \emptyset$
    \ENDFOR
\end{algorithmic}
\vspace*{16pt}

Note that the $i_h$ at line 3 exists because the collection $\mathcal{S}$ is connected.
\end{proof}

\begin{corollary}\label{CorConSize}
 For a connected collection of subsets $\mathcal{S}=\{S_1,\dots,S_t\}$, define an integer $a=\sum_{i=1}^t|S_i|-|\cup_{i=1}^tS_i|$, then $a\geq t-1$.
\end{corollary}
\begin{proof}
By Proposition \ref{PropCon}, we can assume without loss of generality that $\mathcal{S}$ satisfies the condition (\ref{EqConOrder}), i.e., $S_j\cap(\cup_{h=1}^{j-1} S_j)\neq\emptyset \text{ for all } 2 \le j \le t$. Since
\begin{align*}
|\cup_{i=1}^tS_i|&=|S_t|-|S_t\cap(\cup_{i=1}^{t-1} S_i)|+|\cup_{i=1}^{t-1} S_i|\\
&=|S_t|-|S_t\cap(\cup_{i=1}^{t-1} S_i)|+|S_{t-1}|-|S_{t-1}\cap(\cup_{i=1}^{t-2} S_i)|+|\cup_{i=1}^{t-2} S_i|\\
&=\sum_{i=1}^t|S_i|-\sum_{i=2}^{t}|S_i\cap(\cup_{j=1}^{i-1} S_j)|\;,
\end{align*}
We have $a=\sum_{i=2}^{t}|S_i\cap(\cup_{j=1}^{i-1} S_j)|\geq t-1$.
\end{proof}

\begin{remark}
In the following, we introduce a set of integers to characterize the structure of an $(r+1)$-cover. First, for an $(r+1)$-cover $\mathcal{S}=\{S_1,\dots,S_{n_1}\}$, we determine a partition of $[n_1]$, say $[n_1]=I_1\cup\dots\cup I_s$, such that
\begin{itemize}
  \item[(1)] for $1\leq i\leq s$, the induced collection  $\mathcal{S}_{I_i}=\{S_j\mid j\in I_i\}$ is connected; and
  \item[(2)] for $1\leq i<j\leq s$, $(\cup_{h\in I_i}S_h)\cap (\cup_{h\in I_j}S_h)=\emptyset$.
\end{itemize}
In other words, this partition of a collection $\mathcal{S}$ actually corresponds to splitting the graph $G(V,E)$ into connected components, where the graph $G(V,E)$ is determined as in the remark after Definition \ref{DefCon}.
Then for $1\leq i\leq s$, define integers $t_i=|I_i|$ and $a_i=\sum_{j\in I_i}|S_j|-|\cup_{j\in I_i}S_j|$.
\end{remark}
It is easy to derive the following lemma.

\begin{lemma}\label{LemIntCond}
For an $(r+1)$-cover $\mathcal{S}=\{S_1,\dots,S_{n_1}\}$, define integers $s,t_1,\dots,t_s, a_1,\dots,a_s$ as in the above remark. Then the following conditions must hold:
\begin{equation}\label{EqIntCondOut}
\begin{cases}
t_1 +\dots + t_s = n_1; \\
a_1 + \dots + a_s = n_2; \\
a_i \ge t_i -1, \forall 1 \le i \le s;\\
s\geq1; t_i\geq1, \forall 1 \le i \le s.
\end{cases}
\end{equation}
\end{lemma}
\begin{proof}By using the notations in the remark, $I_1\cup\dots\cup I_s$ is a partition of $[n_1]$, therefore
 $a_1 + \dots + a_s=\sum_{i=1}^s(\sum_{j\in I_i}|S_j|-|\cup_{j\in I_i}S_j|)=\sum_{i=1}^{n_1}|S_i|-|\cup_{i=1}^{n_1}S_i|=n_1(r+1)-n=n_2$. The other conditions come from Corollary \ref{CorConSize} and the remark.
\end{proof}

\begin{lemma}\label{LemIntProb}
For any $(r+1)$-cover $\mathcal{S}=\{S_1,\dots,S_{n_1}\}$, define integers $s,t_1,\dots,t_s, a_1,\dots,a_s$ as before, then for $1\leq x \leq n_1$, it holds
$$\Min_{\substack{I\subseteq[n_1] \\ |I|=x}}|\cup_{i\in I}S_i|\leq\Min_{l,h_1,\dots,h_l}(xr+1-\sum_{i=1}^{l-1}(a_{h_i}-t_{h_i}))\;,$$
where the `Min' on the right side is subject to all integers $l,h_1,\dots,h_l$ satisfying
\begin{equation}\label{EqIntCondIn}
t_{h_1}+\dots+t_{h_{l-1}}<x\leq t_{h_1}+\dots+t_{h_{l}}.
\end{equation}
\end{lemma}
\begin{proof}
Suppose $l$ and $h_1,\dots,h_l$ are integers satisfying (\ref{EqIntCondIn}).
Then there exists $J \subseteq I_{h_l}$ such that $|J| = x - (t_{h_1}+ \dots + t_{h_{l-1}})$ and the collection $\mathcal{S}_J$ is connected.
It follows that
\begin{align*}
\Min_{\substack{I \subseteq [n_1] \\ |I| = x}}|\cup_{i \in I} S_i| & \le \sum_{j=1}^{l-1} |\cup_{i \in I_{h_j}} S_i | + |\cup_{i \in J} S_i| \\
& = \sum_{j=1}^{l-1} ( \sum_{i \in I_{h_j}} |S_i| - a_{h_j}) + |\cup_{i \in J} S_i| \\
& \substack{(a) \\ \le} \sum_{j=1}^{l-1} ( \sum_{i \in I_{h_j}} |S_i| - a_{h_j}) + \sum_{i \in J} |S_i| - (|J| -1) \\
& =  \sum_{j=1}^{l-1} (  t_{h_j}(r+1) - a_{h_j}) + |J|(r+1) - (|J| -1)\\
& \substack{(b) \\ =} xr +1 - \sum_{j=1}^{l-1} (a_{h_j} - t_{h_j}),
\end{align*}
where (a) is from Corollary \ref{CorConSize} and (b) is from the equality that $|J| = x - (t_{h_1}+ \dots + t_{h_{l-1}})$.
\end{proof}

\subsection{An Integer Programming Based Bound}\label{SubSecIntProBnd}
In this subsection, we derive an integer programming based bound on the minimum distance of any LRC (Theorem \ref{ThmIntBnd}).
Define
\begin{equation}\label{EqIntPro}
\Psi(x) = \Max_{\substack{ s,t_1,\dots,t_s \\ a_1,\dots,a_s}}  \Min_{\;\;l,h_1,\dots,h_l} ( xr+1 - \sum_{i=1}^{l-1}(a_{h_i} - t_{h_i})),\;\; \forall 1 \le x \le n_1,
\end{equation}
where the `Max' is subject to (\ref{EqIntCondOut}) and the `Min' is subject to (\ref{EqIntCondIn}).
Then the value of $\Psi(x)$ is determined only by integers $n_1$ and $n_2$, or equivalently, by $n$ and $r$.

\begin{theorem}\label{ThmIntBnd}
For any $[n,k,d]$ LRC, it holds $\Phi(x) \le \Psi(x)$ for $1 \le x \le n_1$, and
\begin{equation}\label{EqIntBnd}
d \le n-k+1 - \eta,
\end{equation}
where $\eta = \max\{x : \Psi(x) - x < k\}$.
\end{theorem}
\begin{proof}
First, we show that $\Phi(x) \le \Psi(x),\forall 1 \le x \le n_1$.
By Lemma \ref{LemCov1} and Lemma \ref{LemCov2}, there exists an $(r+1)$-cover $\mathcal{T}$ consisting of $n_1$ subsets $\{T_1,\dots,T_{n_1}\}$ such that
\begin{equation*}
\Phi(x) \le \Min_{\substack{J \in [n_1]\\ |J| = x}} |\cup_{j \in J}T_j|, \forall 0\le x \le n_1.
\end{equation*}
Define integers $s,t_1,\dots,t_s, a_1,\dots,a_s$ as in the remark after Corollary \ref{CorConSize}.
By Lemma \ref{LemIntProb} we have
\begin{equation*}
\Phi(x) \le \Min_{l,h_1,\dots,h_l} (xr+1 - \sum_{i=1}^{l-1} (a_{h_i} - t_{h_i})), \forall 1 \le x \le n_1,
\end{equation*}
where the minimum is subject to (\ref{EqIntCondIn}).
Then it follows from Lemma \ref{LemIntCond} that $\Phi(x) \le \Psi(x)$ for $1 \le x \le n_1$.
Therefore $k > \Psi(\eta) - \eta \ge  \Phi(\eta) - \eta$.
We have $\eta \le \rho$, and then by Theorem \ref{ThmRegDis}, the bound (\ref{EqIntBnd}) is obtained.
\end{proof}

\begin{remark}
{\it Difference between the bound (\ref{EqIntBnd}) and the bound in Theorem \ref{ThmRegDis}.~~} The two bounds are of the same form except that the former is determined by $\eta$ and the function $\Psi(x)$ while the latter is determined by $\rho$ and the function $\Phi(x)$. But $\Psi(x)$ is defined for all integers $n$ and $r$ while $\Phi(x)$ is defined with respect to specific regenerating set structure. In other words, given  parameters $n$ and $r$, the bound (\ref{EqIntBnd}) definitely provide an upper bound for any LRC with the parameters $n$ and $r$, but Theorem \ref{ThmRegDis} cannot give a specific bound due to the lack of information about regenerating set structure. Nevertheless, no efficient algorithm has been established for solving the integer programming problem involved in the bound (\ref{EqIntBnd}). But we can solve it by exhaustive search for small $n$ and $r$ as in the example below. Furthermore, we can determine the solution for a wide class of the values of $n$ and $r$ which plays an important role in practical use. The details are in the next section.
\end{remark}

\begin{example}
Suppose $n=13, r=3$,
then $n_1 = 4$ and $n_2 =3$. Because of the assumption $1< r<k$ and the upper bound on the information rate of LRCs, i.e. $\frac{k}{n}\leq \frac{r}{r+1}$, we consider $4 \le k \le 9$.

First,  compute the value of $\Psi(x)$ for $1 \le x \le 4$.
Observe that, up to permutation, all possible integers $s$ and $\{a_i,t_i\}_{i \in [s]}$ satisfying (\ref{EqIntCondOut}) are

\begin{table}[h]
\centering
\begin{tabular}{||c|c|c||c|c|c||}
\hline \hline
\multirow{2}*{$s=1$} & $t_1$ & $a_1$ & \multirow{4}*{$s=2$} & $(t_1,t_2)$ & $(a_1,a_2)$ \\ \cline{2-3} \cline{5-6}
 & $4$ & $3$ & & $(1,3)$ & $(0,3)$ \\ \cline{1-3} \cline{5-6}
\multirow{6}*{$s=3$} & $(t_1,t_2,t_3)$ & $(a_1,a_2,a_3)$ & & $(1,3)$ & $(1,2)$ \\ \cline{2-3} \cline{5-6}
 & $(1,1,2)$ & $(0,0,3)$ & & $(2,2)$ & $(1,2)$ \\ \cline{2-3} \cline{4-6}
 & $(1,1,2)$ & $(0,1,2)$ & \multirow{3}*{$s=4$} & $(t_1,t_2,t_3,t_4)$ & $(a_1,a_2,a_3,a_4)$ \\ \cline{2-3} \cline{5-6}
 & $(1,1,2)$ & $(0,2,1)$ & & $(1,1,1,1)$ & $(0,0,0,3)$ \\ \cline{2-3} \cline{5-6}
 & $(1,1,2)$ & $(1,1,1)$ & & $(1,1,1,1)$ & $(0,0,1,2)$ \\ \cline{2-3} \cline{5-6}
 & & & & $(1,1,1,1)$ & $(0,1,1,1)$ \\ \hline \hline
\end{tabular}
\end{table}

Then by an exhaustive search, we get $\Psi(1) = 4, \Psi(2)=7,\Psi(3)=10,\Psi(4)=13$.
For simplicity, we can write $\Psi(x) = 3 x +1$ for $ 1 \le x \le 4$.

Therefore we have $\eta = \max\{x : \Psi(x) - x < k\} = \max\{x : 2x+1 < k\} = \left\lceil \frac{k-3}{2} \right\rceil$ for $4 \le k \le 9$.
Thus by Theorem \ref{ThmIntBnd},
\begin{equation}\label{EqExmIntBnd}
d \le n-k+1 - \left\lceil \frac{k-3}{2} \right\rceil.
\end{equation}
It gives an explicit upper bound. We compare it with the well known bound, i.e., the bound (\ref{EqGoplanBnd}) given by Gopolan et al. As displayed in Fig. \nolinebreak \ref{FigExample3}, the bound (\ref{EqExmIntBnd}) goes through three points beneath the bound (\ref{EqGoplanBnd}), i.e. $k=6,9$ and $8$, where the former two points have been expected by the impossible condition $(r+1)\nmid n$ and $r\mid k$ (see Example \ref{ExBnd}) but the point $k=8$ is a new impossible result (not included in the impossible results in \cite{song2014optimal}).

\begin{figure}[tbp]
\centering

\includegraphics[width = 0.52 \textwidth]{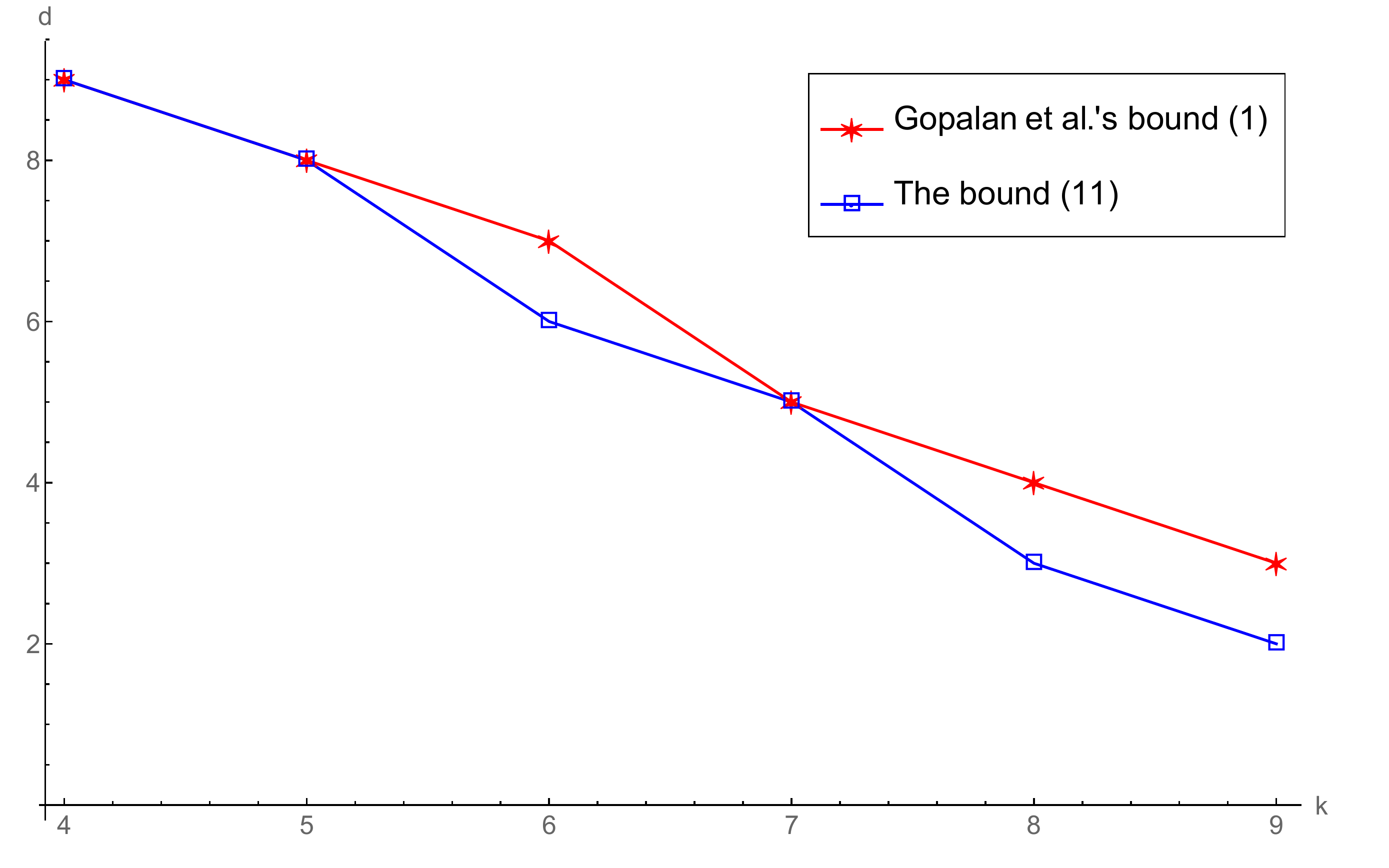}

\caption{Comparison of the two bounds for $n=13, r = 3$.}
\label{FigExample3}
\end{figure}
\end{example}

\section{Explicit Bound for the Case $n_1>n_2$}\label{SecExpBnd}
In this section, for a wide class of parameters, i.e. $n_1>n_2$, we solve the integer programming problem involved in Theorem \ref{ThmIntBnd}, and then derive an explicit upper bound for all LRCs satisfying $n_1>n_2$. Since the condition $n_1>n_2$ can be viewed as a result of $r\leq\sqrt{n}-1$ which is a natural constraint for LRCs to be used in practice, the explicit bound we obtain here is sufficient to cover most practical use. In the second part of this section we make comparisons with all previously known results to show the improvements of our explicit bound.  Actually, in Section \ref{SecCnst} we will show this bound is tight for the case $n_1>n_2$.

\subsection{Bound from Solution of the Integer Programming Problem}
First, Proposition \ref{ProIntN1>N2} determines the value of the function $\Psi(x)$ under the condition $n_1 > n_2$. Then Theorem \ref{ThmBndN1N2} derives an explicit upper bound accordingly.

Denote $\mu = n_1 -n_2$ and let $\lambda, \nu$ be integers such that $n_1 = \lambda \mu +\nu$ and $0 \le \nu < \mu$.

\begin{proposition}\label{ProIntN1>N2}
For $1 \le x \le n_1$,
\begin{equation*}
\Psi(x) = xr+\max\{ \left\lceil \frac{x}{\lambda+1} \right\rceil, \left\lceil \frac{x-\nu}{\lambda} \right\rceil\}.
\end{equation*}
\end{proposition}
\begin{proof}
The proof is given in Appendix \ref{AppIntN1>N2}.
\end{proof}

\begin{theorem}\label{ThmBndN1N2}
For any $[n,k,d]$ LRC with $n_1 > n_2$, where $n_1 = \left\lceil \frac{n}{r+1} \right\rceil$ and $n_2 = n_1(r+1) - n$, it holds
\begin{equation}\label{EqBndN1>N2}
d \le n-k+1-\tilde{\eta},
\end{equation}
where $ \tilde{\eta} = \min \{ \left\lceil \frac{(\lambda+1)(k-1)+1}{(\lambda+1)(r-1) + 1} \right\rceil, \left\lceil \frac{\lambda(k-1) + \nu +1}{\lambda(r-1) + 1} \right\rceil\} -1.$
\end{theorem}
\begin{proof}
We prove this by showing $\tilde{\eta}=\eta$, where $\eta$ is defined in Theorem \ref{ThmIntBnd}. Specifically,
\begin{align*}
\eta & = \max \{x : \Psi(x) -x < k\} \\
& = \max \{x : x(r-1)+\max\{ \left\lceil \frac{x}{\lambda+1} \right\rceil, \left\lceil \frac{x-\nu}{\lambda} \right\rceil \} <k\}\\
& = \max \{x : x(r-1) + \frac{x}{\lambda+1} \le k-1 \text{ and } x(r-1) + \frac{x-\nu}{\lambda} \le k-1 \} \\
& = \max \{x : x \le \frac{(\lambda+1)(k-1)}{(\lambda+1)(r-1) + 1} \text{ and } x \le \frac{\lambda(k-1) + \nu}{\lambda(r-1) + 1}\}.
\end{align*}
Thus we have $ \eta = \min \{ \left\lceil \frac{(\lambda+1)(k-1)+1}{(\lambda+1)(r-1) + 1} \right\rceil, \left\lceil \frac{\lambda(k-1) + \nu +1}{\lambda(r-1) + 1} \right\rceil\} -1=\tilde{\eta}$, and the statement follows directly from Theorem \ref{ThmIntBnd}.
\end{proof}

\begin{example}
Let $\mathcal{C}$ be an $[n,k]$ LRC with $(r+1) \mid n$.
We have $n_1 = \frac{n}{r+1}, n_2 = 0$, and therefore $\mu = n_1, \nu = 0, \lambda = 1$.
Then it follows from Theorem \ref{ThmBndN1N2} that $\tilde\eta = \min \{ \left\lceil \frac{2k-1}{2r-1} \right\rceil, \left\lceil \frac{k}{r} \right\rceil\} -1 = \left\lceil \frac{k}{r} \right\rceil -1$ and
$$ d \le n-k+1 - (\left\lceil \frac{k}{r} \right\rceil -1),$$
which coincides with the bound (\ref{EqGoplanBnd}).
\end{example}

\subsection{Improvements of the Bound}\label{SubSecImprovement}
Since the bound (\ref{EqBndN1>N2}) in Theorem \ref{ThmBndN1N2} holds for $n_1>n_2$, all the comparisons we make below are under the condition $n_1>n_2$.

\subsubsection{Comparison with Gopolan et al's Bound}

The bound (\ref{EqGoplanBnd}) given  by Gopalan et al. \cite{gopalan2012locality} is the first upper bound on the minimum distance of LRCs. It states
$$ d \le n-k+1 - (\left\lceil \frac{k}{r} \right\rceil -1)\;.$$
Because $n_1>n_2$, it follows $\lambda \ge 1$ and $\nu \ge 0$. Then along with the assumption $1<r<k$, a detailed calculation shows that $\tilde{\eta}\geq\left\lceil \frac{k}{r} \right\rceil -1$. Therefore, the bound (\ref{EqBndN1>N2}) generally provides a tighter upper bound than the bound (\ref{EqGoplanBnd}). Actually, the former bound is strictly tighter than the latter at many points. The left graph of Fig. \nolinebreak \ref{FigLocalityR3Bounds} gives a comparison of the two bounds for $n=101,r=9$.

\begin{figure*}[t]
\centering
\subfloat[]{\includegraphics[width=0.5\textwidth]{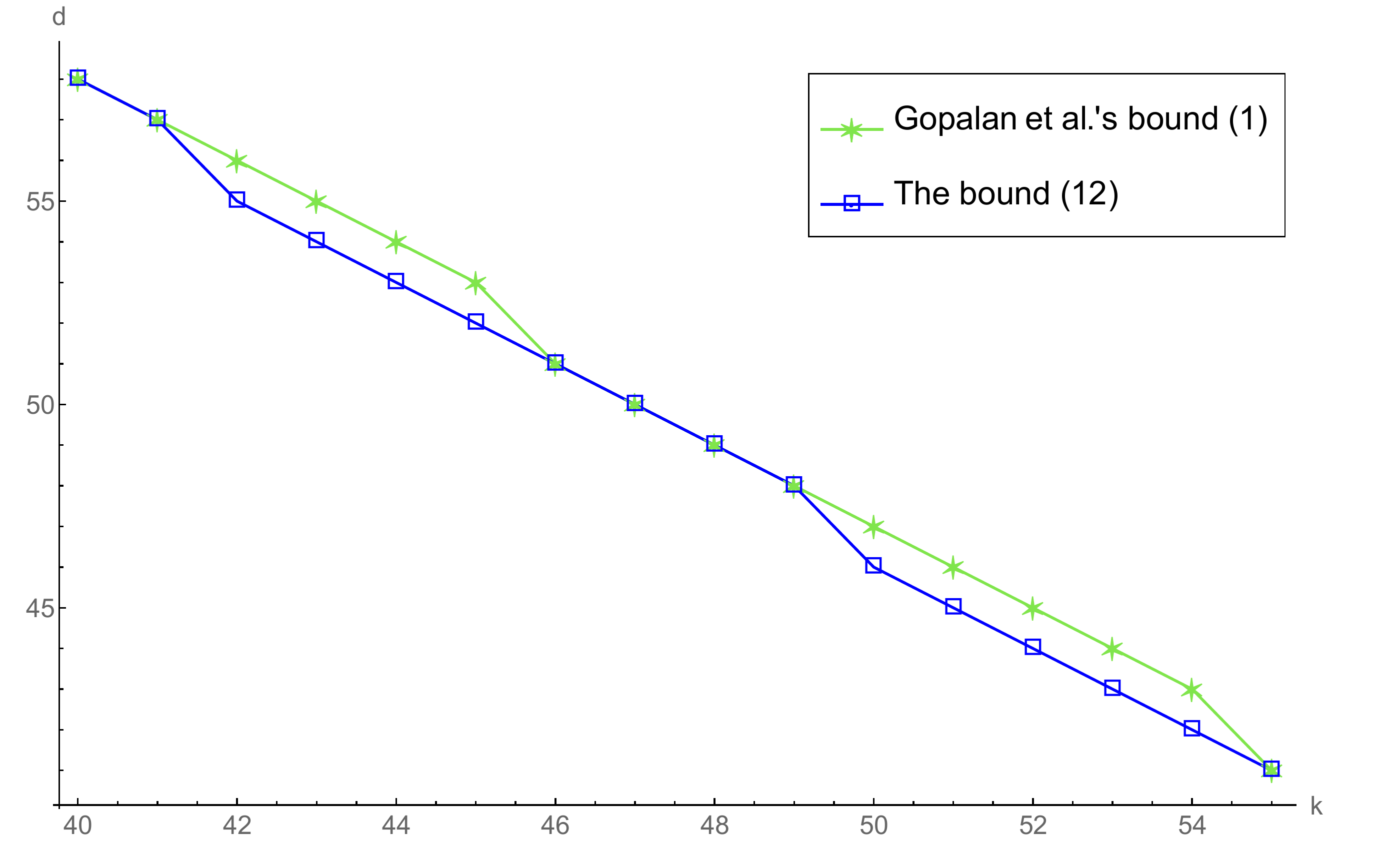}%
\label{FigLocalityR3Bounds_a}}
\hfil
\subfloat[]{\includegraphics[width=0.5\textwidth]{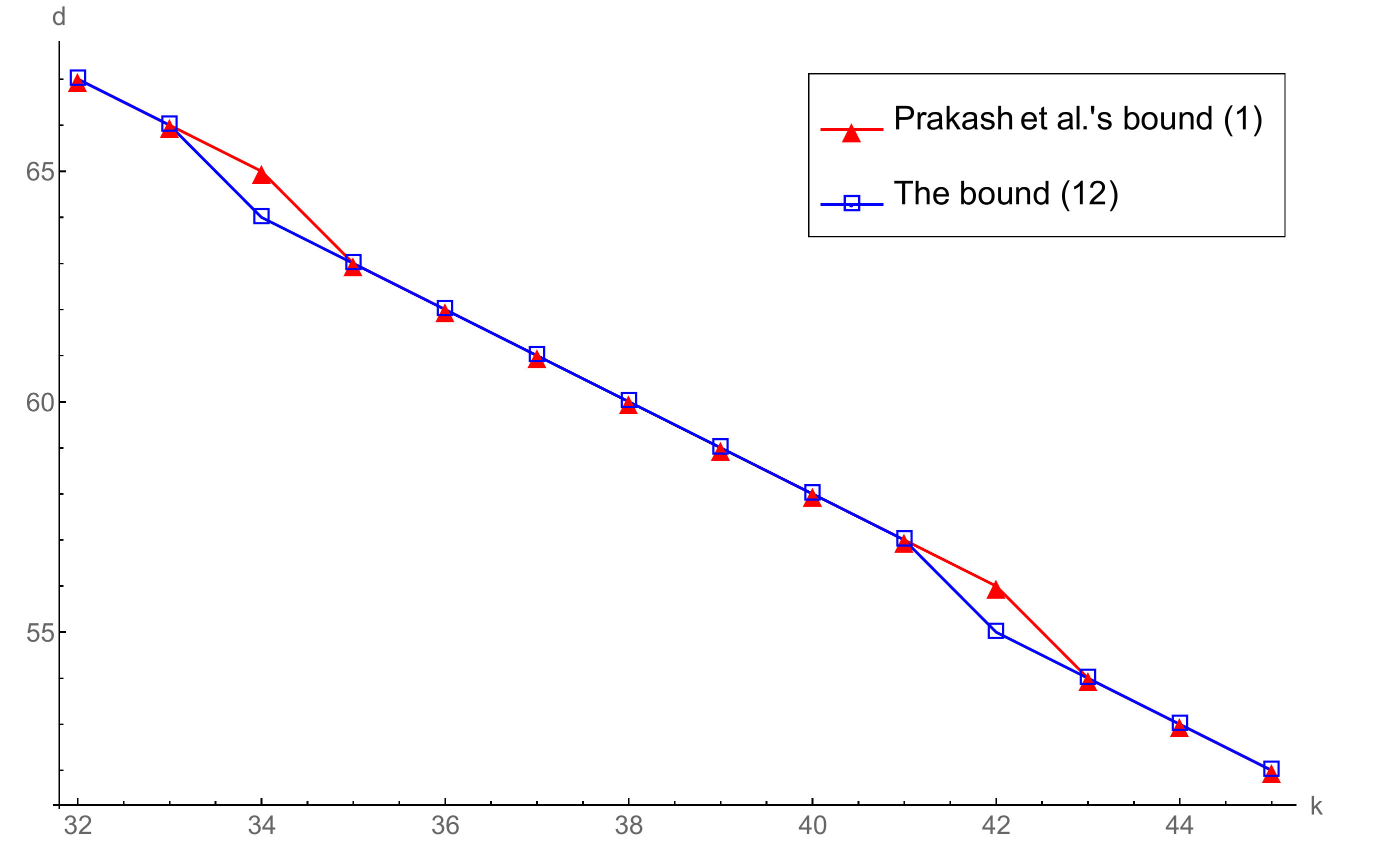}%
\label{FigLocalityR3Bounds_b}}
\caption{Comparison of the three bounds for $n=101, r = 9$.}
\label{FigLocalityR3Bounds}
\end{figure*}

\subsubsection{Comparison with Prakash et al's Bound}
Recently, Prakash et al. \cite{prakash2014codes} derived an improved upper bound on the minimum distance, i.e.,
\begin{equation}\label{EqBndKum}
d \le n-k+1 -l,
\end{equation}
where $l$ is the unique integer satisfying $e_l < k+l < e_{l+1}$ and $\{e_m\}_{m \in [n_1]}$ is defined recursively as below,
$$ e_{n_1} = n \text{~~and~~} e_{m-1} = e_m - \left\lceil \frac{2e_m}{m} \right\rceil +(r+1) \text{ for } 2 \le m \le n_1.$$
It was proved in \cite{prakash2014codes} that the bound (\ref{EqBndKum}) improves the bound (\ref{EqGoplanBnd}). We claim that the bound (\ref{EqBndN1>N2}) in Theorem \ref{ThmBndN1N2} further improves the bound (\ref{EqBndKum}).
Generally, observe that $\eta = \max \{x : \Psi(x) - x <k\}$ and the definition of $l$ is equivalent to $l = \max \{ m | e_m - m < k\}$.
Then the claim follows from the fact that
\begin{equation}\label{EqBndComp}
e_m \ge \Psi(m), \;\;\forall\; 1 \le m \le n_1.
\end{equation}
We prove (\ref{EqBndComp}) by induction on $m$.

First for $m= n_1$, $\Psi(n_1) = n_1 r + \mu = n = e_{n_1}$.
Then suppose the argument holds for $m+1$, i.e., $e_{m+1} \ge \Psi(m+1)$, where $m<n_1$. Thus
\begin{align*}
e_m & = e_{m+1} - \left\lceil \frac{2 e_{m+1}}{m+1} \right\rceil +(r+1) \\
& = \left\lfloor \frac{m-1}{m+1} e_{m+1} \right\rfloor + (r+1) \\
& \ge \left\lfloor \frac{m-1}{m+1} \Psi(m+1) \right\rfloor + (r+1) \\
& = \left\lfloor \frac{m-1}{m+1} ((m+1)r+\max\{ \left\lceil \frac{m+1}{\lambda+1} \right\rceil, \left\lceil \frac{m+1-\nu}{\lambda} \right\rceil\}) \right\rfloor + (r+1) \\
& = mr+1 + \left\lfloor \frac{m-1}{m+1} \max\{ \left\lceil \frac{m+1}{\lambda+1} \right\rceil, \left\lceil \frac{m+1-\nu}{\lambda} \right\rceil\} \right\rfloor \\
& \ge mr+1 + \left\lfloor \max \{  \frac{m-1}{\lambda+1} , \frac{m-1-\nu}{\lambda} \} \right\rfloor \\
& = mr +1 + \max \{ \left\lceil \frac{m}{\lambda+1} \right\rceil -1, \left\lceil \frac{m-\nu}{\lambda} \right\rceil -1 \} \\
& = \Psi(m).
\end{align*}
The above proof shows that the bound (\ref{EqBndN1>N2}) cannot go upon the bound (\ref{EqBndKum}). A detailed calculation with specific values of $n,k,r$ shows the former bound does go beneath the latter bound at some points.
As an illustration, the right graph in Fig. \nolinebreak \ref{FigLocalityR3Bounds} plots the two bounds for $n =101, r =9$.

\textbf{}

\subsubsection{Comparing with the Results of Song et al }
In \cite{song2014optimal}, Song et al. derived some conditions under which there exists {\it no} LRC attaining the bound (\ref{EqGoplanBnd}), and also proved the existence of LRCs attaining the bound (\ref{EqGoplanBnd}) under some conditions.
However, they left some scope of parameters under which it was unknown whether there exist LRCs attaining the bound (\ref{EqGoplanBnd}).

In Section \ref{SecCnst} of this paper, we will give an explicit construction of $[n,k]$ LRCs for $n_1 > n_2$, attaining the bound (\ref{EqBndN1>N2}) in Theorem \ref{ThmBndN1N2}.
Therefore our bound (\ref{EqBndN1>N2}) completely describes the largest possible minimum distance for LRCs with $n_1 > n_2$.

Fig. \nolinebreak \ref{FigCompSong} illustrates the corresponding results for $n=50,10 \le k \le 17$ and $2 \le r \le 9$. In the tables `Y' means there exist LRCs attaining the bound (\ref{EqGoplanBnd}),  `N' means there is {\it no} LRC attaining the bound (\ref{EqGoplanBnd}), and a blank means it is unknown whether there exist LRCs attaining the bound (\ref{EqGoplanBnd}).

\begin{figure}[ht]
\centering
\begin{tikzpicture}[scale=0.8]

\foreach \x / \y in {1/1,2/1,3/1,4/1,5/1,6/1,7/1,8/1, 1/2,2/2,3/2,4/2,5/2,8/2, 1/3,2/3,3/3,6/3,7/3,8/3, 1/4,3/4,4/4,5/4, 2/5,3/5,4/5,5/5,7/5,8/5, 1/6,2/6,3/6,4/6,5/6,6/6,7/6,8/6, 1/7,4/7}
{
\filldraw [green!30] (\x, 0.7*\y-0.7) rectangle (\x+1,0.7*\y);
\node at (\x+.5,0.7*\y-0.7*.5) {\small Y};
}

\foreach \x / \y in {7/2,5/3,3/4, 1/5,6/5, 2/7,3/7,5/7,6/7,8/7, 1/8,3/8,5/8,7/8}
{
\filldraw [red!30] (\x, 0.7*\y-0.7) rectangle (\x+1,0.7*\y);
\node at (\x+.5,0.7*\y-0.7*.5) {\small N};
}

\foreach \x in {0,1,2,3,4,5,6,7,8,9}
{
\draw[line width=1pt] (0,0.7*\x) -- (9,0.7*\x);
\draw[line width=1pt] (\x,0) -- (\x,0.7*9);
}

\draw (0,0.7*9) -- (1,0.7*8);
\node [above right] at (0.5,0.7*8.5-0.2) {\small $k$};
\node [below left] at (0.5,0.7*8.5+0.1) {\small $r$};

\foreach \x in {10,11,12,13,14,15,16,17}
{
\node at (\x-9+0.5,0.7*8.5) {\small $\x$};
}

\foreach \x in {2,3,4,5,6,7,8,9}
{
\node at (0.5,0.7*9.5-0.7*\x) {\small $\x$};
}

\node at (4.5,6.8) {\small Results of Song et al.};


\foreach \x / \y in {1/1,2/1,3/1,4/1,5/1,6/1,7/1,8/1, 1/2,2/2,3/2,4/2,5/2,6/2,8/2, 1/3,2/3,3/3,4/3,6/3,7/3,8/3, 1/4,2/4,4/4,5/4,6/4,7/4, 2/5,3/5,4/5,5/5,7/5,8/5, 1/6,2/6,3/6,4/6,5/6,6/6,7/6,8/6, 1/7,4/7,7/7, 2/8,4/8,6/8,8/8}
{
\filldraw [xshift=10cm,green!30] (\x, 0.7*\y-0.7) rectangle (\x+1,0.7*\y);
\node at (\x+.5+10,0.7*\y-0.7*.5) {\small Y};
}

\foreach \x / \y in {7/2,5/3,3/4,8/4, 1/5,6/5, 2/7,3/7,5/7,6/7,8/7, 1/8,3/8,5/8,7/8}
{
\filldraw [xshift=10cm,red!30] (\x, 0.7*\y-0.7) rectangle (\x+1,0.7*\y);
\node at (\x+.5+10,0.7*\y-0.7*.5) {\small N};
}

\foreach \x in {0,1,2,3,4,5,6,7,8,9}
{
\draw[xshift=10cm,line width=1pt] (0,0.7*\x) -- (9,0.7*\x);
\draw[xshift=10cm,line width=1pt] (\x,0) -- (\x,0.7*9);
}

\draw [xshift=10cm] (0,0.7*9) -- (1,0.7*8);
\node [above right] at (0.5+10,0.7*8.5-0.2) {\small $k$};
\node [below left] at (0.5+10,0.7*8.5+0.1) {\small $r$};

\foreach \x in {10,11,12,13,14,15,16,17}
{
\node at (\x-9+0.5+10,0.7*8.5) {\small $\x$};
}

\foreach \x in {2,3,4,5,6,7,8,9}
{
\node at (0.5+10,0.7*9.5-0.7*\x) {\small $\x$};
}

\node at (4.5+10,6.8) {\small Results of this paper};

\end{tikzpicture}
\caption{A comparison of Song et al.'s results and our results for $n=50$.}
\label{FigCompSong}
\end{figure}

\section{Code Construction When $n_1>n_2$}\label{SecCnst}
In this section, we present an explicit construction of LRCs attaining the bound (\ref{EqBndN1>N2}).
The construction is based on linearized polynomials.
We start this section with some basic facts about linearized polynomials.

\subsection{The Linearized Polynomial}
\begin{definition} A polynomial of the form $f(x) = \sum_{i=0}^t a_i x^{q^i}$ with coefficients $a_i \in \mathbb{F}_{q^m}$ for $0\leq i\leq t$ and $a_t \neq 0$ is called a linearized polynomial of $q$-degree $t$ over the extension field $\mathbb{F}_{q^m}$.
\end{definition}

A linearized polynomial $f(x)$ can be viewed as an $\mathbb{F}_q$-linear transformation from $\mathbb{F}_{q^m}$ to itself, i.e., for any $c_1,c_2 \in \mathbb{F}_q$ and $\omega_1,\omega_2 \in \mathbb{F}_{q^m}$, it holds $f(c_1 \omega_1 + c_2 \omega_2) = c_1 f(\omega_1) + c_2 f(\omega_2)$.
Furthermore, a standard result of finite fields states that,

\begin{proposition}{\upshape \cite{lidl1997finite}}\label{PropFiniteField}
A linearized polynomial $f(x)$ of $q$-degree no more than $t$ can be uniquely determined by the values of $f(\omega_1),\dots,f(\omega_{t+1})$, where $\omega_1,\dots,\omega_{t+1}$ are $t+1$ elements in $\mathbb{F}_{q^m}$ that are linearly independent over $\mathbb{F}_q$.
\end{proposition}

\subsection{An Explicit Code Construction}\label{SubSecExplicitCnst}
In this subsection, we assume $n_1 > n_2$ and construct an $[n,k]$ LRC over $\mathbb{F}_{q^m}$ attaining the bound (\ref{EqBndN1>N2}) in Theorem \ref{ThmBndN1N2}, where $\mathbb{F}_{q^m}$ is an extension field of $\mathbb{F}_{q}$ with $m \ge n_1 r$.
In a word, the codewords are obtained as evaluations of a linearized polynomial at $n$ points in $\mathbb{F}_{q^m}$. Because of the property of linearized polynomials introduced in Proposition \ref{PropFiniteField}, the key point of the code construction is the selection of the $n$ evaluation points such that the resulting code has the largest possible minimum distance. Denote the set of the $n$ evaluation points by $\Omega$.

Since $\mathbb{F}_{q^m}$ can be viewed as an $\mathbb{F}_q$-linear space of dimension $m$, by fixing a basis of $\mathbb{F}_{q^m}$ over $\mathbb{F}_q$, the $n$ elements in $\Omega$ can be expressed as $n$ vectors of length $m$ over $\mathbb{F}_q$.
These $n$ vectors are determined through the following three steps.
For simplicity, we can set $q=2$ and $m=n_1 r$, and the process below also works for other values of $q$ and $m$.

{\it Step 1.} Let $X=(\vec{x}_0,\vec{x}_1,\dots,\vec{x}_r)$ be the generator matrix of an $[r+1,r]_2$ MDS code and let $\vec{c}=(1,c_1,\dots,c_r)$ be one of its codeword, where $\vec{x}_i \in \mathbb{F}_{2}^{\;r}$ for $0 \le i \le r$.
For example, we can choose
\begin{equation*}
X =
\begin{pmatrix}
1 & 0 & \dots & 0 & 1 \\
0 & 1 & \dots & 0 & 1 \\
\vdots & \vdots & \ddots & \vdots \\
0 & 0 & \cdots & 1 & 1
\end{pmatrix}
\text{ and } \vec{c} = (1,0,\dots,0,1).
\end{equation*}

{\it Step 2.} Define vectors $\vec{\alpha}_0,\vec{\alpha}_{i,j} \in \mathbb{F}_2^{(\lambda+1)r}$ for $1 \le i \le \lambda+1$ and $1 \le j \le r$, where
\begin{equation*}
\vec{\alpha}_0 = \begin{pmatrix}\vec{x}_0 \\ \vec{x}_0 \\ \vdots \\ \vec{x}_0 \end{pmatrix} \text{ and } \vec{\alpha}_{i,j} = \begin{pmatrix}c_j\vec{x}_0 \\ \vdots \\ \vec{x}_j \\ \vdots \\ c_j\vec{x}_0 \end{pmatrix}\;,
\end{equation*}
that is, $\vec{\alpha}_0$ consists of $(\lambda+1)$ $\vec{x}_0$'s and $\vec{\alpha}_{i,j}$ is defined by replacing the $i$-th $c_j \vec{x}_0$ of $c_j \vec{\alpha}_0$ with an $\vec{x}_j$.
Similarly, define $\vec{\beta}_0,\vec{\beta}_{i,j} \in \mathbb{F}_2^{\lambda r}$, $1 \le i \le \lambda$ and $1 \le j \le r$, such that $\vec{\beta}_0$ consists of $\lambda$ $\vec{x}_0$'s and $\vec{\beta}_{i,j}$ is defined by replacing the $i$-th $c_j \vec{x}_0$ of $c_j \vec{\alpha}_0$ with an $\vec{x}_j$.
For example, let $r=2,\lambda=2$ and $X=(\vec{x}_0,\vec{x}_1,\vec{x}_2)$, $\vec{c}=(1,c_1,c_2)$, then we have
\begin{align*}
(\vec{\alpha}_0, \vec{\alpha}_{1,1}, \vec{\alpha}_{1,2}, \vec{\alpha}_{2,1}, \vec{\alpha}_{2,2}, \vec{\alpha}_{3,1}, \vec{\alpha}_{3,2} )=
\begin{pmatrix}
\vec{x}_0 & \vec{x}_1 & \vec{x}_2 & c_1\vec{x}_0 & c_2\vec{x}_0 & c_1\vec{x}_0 & c_2\vec{x}_0 \\
\vec{x}_0 & c_1\vec{x}_0 & c_2\vec{x}_0 & \vec{x}_1 & \vec{x}_2 & c_1\vec{x}_0 & c_2\vec{x}_0 \\
\vec{x}_0 & c_1\vec{x}_0 & c_2\vec{x}_0 & c_1\vec{x}_0 & c_2\vec{x}_0 & \vec{x}_1& \vec{x}_2
\end{pmatrix}
\end{align*}
and
\begin{align*}
(\vec{\beta}_0, \vec{\beta}_{1,1}, \vec{\beta}_{1,2}, \vec{\beta}_{2,1}, \vec{\beta}_{2,2})=
\begin{pmatrix}
\vec{x}_0 & \vec{x}_1 & \vec{x}_2 & c_1\vec{x}_0 & c_2\vec{x}_0 \\
\vec{x}_0 & c_1\vec{x}_0 & c_2\vec{x}_0 & \vec{x}_1 & \vec{x}_2
\end{pmatrix}.
\end{align*}
The vectors $\vec{\alpha}_0,\vec{\alpha}_{i,j}$ and $\vec{\beta}_0,\vec{\beta}_{i,j}$ defined above have the following properties.

\begin{lemma}\label{LemAlphaBeta}
Denote $\mathcal{A}_i = \{\vec{\alpha}_0,\vec{\alpha}_{i,1},\dots,\vec{\alpha}_{i,r}\}$ for $1 \le i \le \lambda+1$.
Then we have
\begin{itemize}
\item[(i)] For $1 \le i \le \lambda+1$, each vector contained in $\mathcal{A}_i$ is an $\mathbb{F}_q$-linear combination of the other $r$  vectors in $\mathcal{A}_i$.
\item[(ii)] For any $F \subseteq \cup_{i=1}^{\lambda+1} \mathcal{A}_i$ satisfying that there exists a vector $\vec{\alpha} \in F$ such that $|(F - \{\vec{\alpha}\}) \cap \mathcal{A}_i| \le r-1$ for $1 \le i \le \lambda+1$, the vectors in $F$ are $\mathbb{F}_q$-linearly independent.
\end{itemize}
Denote $\mathcal{B}_j = \{\vec{\beta}_0,\vec{\beta}_{j,1},\dots,\vec{\beta}_{j,r}\}$ for $1 \le j \le \lambda$.
Then the same statements also hold for $\mathcal{B}_j$ for $1 \le j \le \lambda$.
\end{lemma}
\begin{proof}
The proof is given in Appendix \ref{AppPrfLemAlpBet}.
\end{proof}

{\it Step 3.} Let $A$ be the matrix consisting of the $((\lambda+1)r+1)$ column vectors in $\cup_{i=1}^{\lambda+1} \mathcal{A}_i$, and let $B$ be the matrix consisting of the $(\lambda r+1)$ column vectors in $\cup_{i=1}^{\lambda} \mathcal{B}_i$.
Define a block diagonal matrix
\begin{equation*}
W= \begin{pmatrix}
A & & & & & & \\
 &\ddots & & & & & \\
 & & A & & & & \\
 & & & B & & & \\
 & & & & \ddots & \\
 & & & & & B
\end{pmatrix}
\end{equation*}
which is composed of $\nu$ $A$'s and $(\mu-\nu)$ $B$'s on the diagonal and zeros eleswhere.
Note that $A$ has $((\lambda+1)r +1)$ columns and $B$ has $(\lambda r +1)$ columns, then $W$ has $((\lambda+1)r +1)\nu + (\lambda r +1)(\mu-\nu) = n_1 (r+1) - n_2 = n$ columns.
Similarly, $W$ has $\nu(\lambda+1)r +(\mu-\nu)\lambda r =(\lambda \mu + \nu)r = n_1 r$ rows.
Then the set of $n$ vectors in $\Omega$ are defined to be the $n$ columns of $W$.

We give a graphical explanation of linear dependences among the $n$ vectors. Refer to Fig. \nolinebreak \ref{FigOmegaStructure}, each point actually corresponds to a vector. Then the left $\nu$ trees each composed of $\lambda+1$ branches corresponds to the $\nu$ blocks of $A$ in $W$, and the right $\mu-\nu$ trees each composed of $\lambda$ branches corresponds to the $\mu-\nu$ blocks of $B$ in $W$. In more detail, the set $\mathcal{A}_i$ for $1\leq i\leq \lambda+1$ corresponds to a branch in the left trees and particularly the vector $\vec{\alpha}_0$ corresponds to the root point. The similar correspondence holds for $\mathcal{B}_i$ and the branches in the right trees.

\begin{figure}[ht]
\centering
\begin{tikzpicture}[scale=1]

\filldraw[blue]
(0,0) circle (1pt);

\foreach \x in {1,1/4,2/4,3/4}
{\filldraw[blue] (-\x,-3 * \x) circle (1pt);
\filldraw[blue] (\x,-3 * \x) circle (1pt);
\filldraw[blue] (-3.1/6 * \x,-3.1 * \x) circle (1pt);}

\node [above left] at (-0.7/4,-3.7/4) { \tiny $\omega^{(1)}_{1,1}$};
\node [left] at (-0.9,-3) { \tiny $\omega^{(1)}_{1,r}$};
\node [below] at (-3.1/24 + 0.16,-3.1/4) {\tiny $\omega^{(1)}_{2,1}$};
\node [right] at (-3.1/6,-3.1) {\tiny $\omega^{(1)}_{2,r}$};
\node [right] at (1/4,-3/4) {\tiny $\omega^{(1)}_{\lambda+1,1}$};
\node [right] at (1,-3) {\tiny $\omega^{(1)}_{\lambda+1,r}$};

\draw
(0,0) -- (-1,-3)
(0,0) -- (1,-3)
(0,0) -- (-3.1/6,-3.1);

\node at (.2,-2) {$\cdots$};
\node [above] at (0,0) {\tiny $\omega^{(1)}_0$};
\node at (-1.2,-3.5) {\footnotesize $W^{(1)}_1$};
\node at (-0.4,-3.65) {\footnotesize $W^{(1)}_2$};
\node at (1.1,-3.5) {\footnotesize $W^{(1)}_{\lambda+1}$};

\node at (1.8,-1.2) {\large $\cdots$};
\filldraw[blue]
(0+3.6,0) circle (1pt);

\foreach \x in {1,1/4,2/4,3/4}
{\filldraw[blue] (-\x+3.6,-3 * \x) circle (1pt);
\filldraw[blue] (\x+3.6,-3 * \x) circle (1pt);
\filldraw[blue] (-3.1/6 * \x+3.6,-3.1 * \x) circle (1pt);}

\draw
(0+3.6,0) -- (-1+3.6,-3)
(0+3.6,0) -- (1+3.6,-3)
(0+3.6,0) -- (-3.1/6+3.6,-3.1);

\node at (.2+3.6,-2) {$\cdots$};
\node [above] at (0+3.6,0) {\tiny $\omega^{(\nu)}_0$};
\node at (-1.2+3.6,-3.5) {\footnotesize $W^{(\nu)}_1$};
\node at (-0.4+3.6,-3.65) {\footnotesize $W^{(\nu)}_2$};
\node at (1.1+3.6,-3.5) {\footnotesize $W^{(\nu)}_{\lambda+1}$};

\node at (1.7,-4.2) {\tiny $\underbrace{\hspace*{110pt}}$};
\node at (1.7,-4.8) {\footnotesize $\nu$ trees each composed of $\lambda+1$ branches};
\filldraw[blue]
(0+3.6+3.6,0) circle (1pt);

\foreach \x in {1,1/4,2/4,3/4}
{\filldraw[blue] (-\x+3.6+3.6,-3 * \x) circle (1pt);
\filldraw[blue] (\x+3.6+3.6,-3 * \x) circle (1pt);
\filldraw[blue] (-3.1/6 * \x+3.6+3.6,-3.1 * \x) circle (1pt);}

\draw
(0+3.6+3.6,0) -- (-1+3.6+3.6,-3)
(0+3.6+3.6,0) -- (1+3.6+3.6,-3)
(0+3.6+3.6,0) -- (-3.1/6+3.6+3.6,-3.1);

\node at (.2+3.6+3.6,-2) {$\cdots$};
\node [above] at (0+3.6+3.6,0) {\tiny $\omega^{(\nu+1)}_0$};
\node at (-1.2+3.6+3.6,-3.5) {\footnotesize $W^{(\nu+1)}_1$};
\node at (-0.4+3.6+3.6,-3.65) {\footnotesize $W^{(\nu+1)}_2$};
\node at (1.1+3.6+3.6,-3.5) {\footnotesize $W^{(\nu+1)}_{\lambda}$};

\node at (1.8+3.6+3.6,-1.2) {\large $\cdots$};
\filldraw[blue]
(0+3.6+3.6+3.6,0) circle (1pt);

\foreach \x in {1,1/4,2/4,3/4}
{\filldraw[blue] (-\x+3.6+3.6+3.6,-3 * \x) circle (1pt);
\filldraw[blue] (\x+3.6+3.6+3.6,-3 * \x) circle (1pt);
\filldraw[blue] (-3.1/6 * \x+3.6+3.6+3.6,-3.1 * \x) circle (1pt);}

\draw
(0+3.6+3.6+3.6,0) -- (-1+3.6+3.6+3.6,-3)
(0+3.6+3.6+3.6,0) -- (1+3.6+3.6+3.6,-3)
(0+3.6+3.6+3.6,0) -- (-3.1/6+3.6+3.6+3.6,-3.1);

\node at (.2+3.6+3.6+3.6,-2) {$\cdots$};
\node [above] at (0+3.6+3.6+3.6,0) {\tiny $\omega^{(\mu)}_0$};
\node at (-1.2+3.6+3.6+3.6,-3.5) {\footnotesize $W^{(\mu)}_1$};
\node at (-0.4+3.6+3.6+3.6,-3.65) {\footnotesize $W^{(\mu)}_2$};
\node at (1.1+3.6+3.6+3.6,-3.5) {\footnotesize $W^{(\mu)}_{\lambda}$};
\node at (1.7+3.6+3.6,-4.2) {\tiny $\underbrace{\hspace*{110pt}}$};
\node at (1.7+3.6+3.6,-4.8) {\footnotesize $\mu-\nu$ trees each composed of $\lambda$ branches};

\end{tikzpicture}
\caption{The $n$ points in $\Omega$.}
\label{FigOmegaStructure}
\end{figure}

For convenience, we denote the $n$ points (or equivalently, the $n$ vectors in $\Omega$) by $$\{\omega^{(l)}_0,\omega^{(l)}_{i,j}\mid l \in[\mu],i \in [\lambda+1], j \in [r]\}\footnote{A precise description of the range of $l$ and $i,j$ is $(l,i) \in ([\nu] \times [\lambda+1])\cup([\nu+1,\mu]\times [\lambda])$, $j \in [r]$, where $[\nu+1,\mu] = \{\nu+1,\nu+2,\dots,\mu\}$. }$$
where the superscript $l$ denotes which tree it belongs to, the subscript $i$ denotes which branch it lies in and $j$ is the point index in that branch.
Moveover, denote each branch by
$$W^{(l)}_i = \{\omega^{(l)}_0, \omega^{(l)}_{i,1}, \omega^{(l)}_{i,2}, \dots, \omega^{(l)}_{i,r}\} \text{ for } l \in [\mu] \text{ and } i \in [\lambda+1].$$
Then by Lemma \ref{LemAlphaBeta} (i), each vector in $W^{(l)}_i$ is an $\mathbb{F}_q$-linear combination of the other $r$ vectors in $W^{(l)}_i$,
and by the construction of the matrix $W$, the vectors in different trees are linearly independent.

\begin{construction}\label{Cnst}
Define an $[n,k]$ linear code $\mathcal{C}$ over $\mathbb{F}_{q^m}$ as follows.
\begin{itemize}
\item Let $\Omega \subseteq \mathbb{F}_{q^m}$ be s set of the $n$ vectors defined above, i.e., $\Omega=\{\omega^{(l)}_0,\omega^{(l)}_{i,j}\mid l \in[\mu],i \in [\lambda+1], j \in [r]\}$. Note that each vector is of length $n_1r=m$ over $\mathbb{F}_q$ and thus can be viewed as an element in $\mathbb{F}_{q^m}$.

\item $\mathcal{C}$ encodes a file $(m_0,\dots,m_{k-1}) \in \mathbb{F}_{q^m}^k$ into $(f(\omega))_{\omega \in \Omega} \in \mathbb{F}_{q^m}^n$, where $f(x) = \sum_{i=0}^{k-1} m_i x^{q^i}$.
\end{itemize}
Denote the $n$ coordinates of $\mathcal{C}$ by the corresponding element in $\Omega$,
then $W^{(l)}_i$ is a regenerating set of each coordinate contained in $W^{(l)}_i$.
Therefore, $\mathcal{C}$ is an $[n,k]$ LRC with locality $r$.
\end{construction}

\begin{example}
We illustrate the construction through a specific example. Suppose $n=8, k=4, r=2$, then it has $n_1=3,n_2=1$ and $\lambda=1,\mu=2,\nu=1$.

The construction is over the field $\mathbb{F}_{2^6} = \mathbb{F}_2(\theta)$, where $\theta$ is a primitive element of $\mathbb{F}_{2^6}$ with minimal polynomial $x^6+x^5+1$.
By fixing a basis $\{1,\theta,\theta^2,\dots,\theta^5\}$, the subset $\Omega \subseteq \mathbb{F}_{2^6}$ is constructed as follows.

First, let
\begin{equation*}
X=(\vec{x}_0,\vec{x}_1,\vec{x}_2) = \begin{pmatrix} 1 & 0  & 1 \\ 0 & 1  & 1 \end{pmatrix}
\text{ and }
\vec{c}=(1,c_1,c_2) = (1,0,1).
\end{equation*}
Then
\begin{align*}
A =& (\vec{\alpha}_0,\vec{\alpha}_{1,1},\vec{\alpha}_{1,2},\vec{\alpha}_{2,1},\vec{\alpha}_{2,2}) \\
= &
\begin{pmatrix}
\vec{x}_0 & \vec{x}_1 & \vec{x}_2 & c_1\vec{x}_0 & c_2\vec{x}_0 \\
\vec{x}_0 & c_1\vec{x}_0 & c_2\vec{x}_0 & \vec{x}_1 & \vec{x}_2
\end{pmatrix}
=
\begin{pmatrix}
1 & 0 & 1 & 0 & 1 \\
0 & 1 & 1 & 0 & 0 \\
1 & 0 & 1 & 0 & 1 \\
0 & 0 & 0 & 1 & 1
\end{pmatrix}
\\
\text{ and }
B =& (\vec{\beta}_0,\vec{\beta}_{1,1},\vec{\beta}_{1,2}) \\
= & (\vec{x}_0, \vec{x}_1,\vec{x}_2) = \begin{pmatrix} 1 & 0 & 1 \\ 0 & 1 & 1 \end{pmatrix}.
\end{align*}
Therefore
\begin{equation*}
W = \begin{pmatrix} A & \\ & B \end{pmatrix} =
\begin{pmatrix}
1 & 0 & 1 & 0 & 1 & 0 & 0 & 0 \\
0 & 1 & 1 & 0 & 0 & 0 & 0 & 0 \\
1 & 0 & 1 & 0 & 1 & 0 & 0 & 0 \\
0 & 0 & 0 & 1 & 1 & 0 & 0 & 0 \\
0 & 0 & 0 & 0 & 0 & 1 & 0 & 1 \\
0 & 0 & 0 & 0 & 0 & 0 & 1 & 1
\end{pmatrix},
\end{equation*}
and thus
\begin{align*}
\begin{matrix}
\omega^{(1)}_0 = &1\hphantom{{}+\theta}+\theta^2\hphantom{{}+\theta^3}\\
\omega^{(1)}_{1,1} =&\hphantom{1+{}}\theta \hphantom{{}+\theta^2+\theta^3}\\
\omega^{(1)}_{1,2} =&1+\theta+\theta^2 \hphantom{{}+\theta^3}\\
\omega^{(1)}_{2,1} =&\hphantom{1 +\theta+\theta^2+{}}\theta^3\\
\omega^{(1)}_{2,2} =&1\hphantom{\theta+{}}+\theta^2+\theta^3
\end{matrix}
\text{\hspace*{16pt} and \hspace*{16pt}}
\begin{matrix}
\omega^{(2)}_0=& \theta^4 \hphantom{{}+\theta^5} \\
\omega^{(2)}_{1,1}=& \hphantom{\theta^4+{}}\theta^5 \\
\omega^{(2)}_{1,2}=& \theta^4+\theta^5
\end{matrix}.
\end{align*}
Fig. \nolinebreak \ref{FigOmegaStructureExample} gives a graphical illustration of the eight elements in $\Omega$.

\begin{figure}
\centering
\begin{tikzpicture}[scale=1]

\filldraw[blue]
(0,0) circle (1pt);

\foreach \x in {1,1/2}
{\filldraw[blue] (-\x,-3 * \x) circle (1pt);
\filldraw[blue] (\x,-3 * \x) circle (1pt);}

\node [above left] at (-0.9/2,-3.7/2) { \tiny $\omega^{(1)}_{1,1}$};
\node [left] at (-0.9,-3) { \tiny $\omega^{(1)}_{1,2}$};
\node [right] at (1/2,-3/2) {\tiny $\omega^{(1)}_{2,1}$};
\node [right] at (1,-3) {\tiny $\omega^{(1)}_{2,2}$};

\draw
(0,0) -- (-1,-3)
(0,0) -- (1,-3);

\node [above] at (0,0) {\tiny $\omega^{(1)}_0$};
\node at (-1.2,-3.5) {\footnotesize $W^{(1)}_1$};
\node at (1.1,-3.5) {\footnotesize $W^{(1)}_{2}$};


\foreach \x in {1,1/2,0}
{
\filldraw[blue] (3,-3 * \x) circle (1pt);
}

\node [left] at (3, -3 * 1) { \tiny $\omega^{(2)}_{1,2}$};
\node [left] at (3, -3 * 1/2) { \tiny $\omega^{(2)}_{1,1}$};
\node [left] at (3, -3 * 0) { \tiny $\omega^{(2)}_{0}$};

\draw (3,0) -- (3,-3);

\node at (3,-3.5) {\footnotesize $W^{(2)}_1$};

\end{tikzpicture}
\caption{The eight elements of $\Omega$ for the $[8,4]$ code.}
\label{FigOmegaStructureExample}
\end{figure}
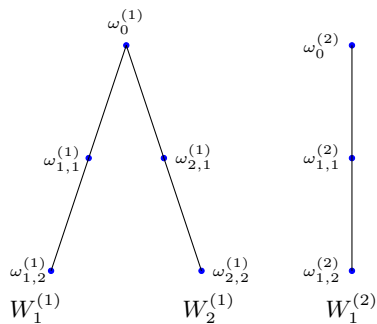

Then  the $[n=8,k=4]$ linear code $\mathcal{C}$ encodes a file $(m_0,m_1,m_2,m_3)$ into $(f(\omega))_{\omega \in \Omega}$, where $f(x) = m_0 x + m_1 x^2 + m_2 x^4 + m_3 x^8$.

A sequence of regenerating sets of the linear code $\mathcal{C}$ is
\begin{equation*}
\{\omega^{(1)}_0, \omega^{(1)}_{1,1}, \omega^{(1)}_{1,2}\}, \{\omega^{(1)}_0,\omega^{(1)}_{2,1}, \omega^{(1)}_{2,2}\}, \{\omega^{(2)}_0, \omega^{(2)}_{1,1}, \omega^{(2)}_{1,2}\},
\end{equation*}
and it is easy to see that $\Phi(1) = 3,\Phi(2)=5,\Phi(3)=8$, which coincides with  the upper bound defined by $\Psi(x)$ (see Proposition \ref{ProIntN1>N2}). Moreover, it can be verified that the minimum distance of $\mathcal{C}$ is $d = 3$, which is optimal with respect to the bound (\ref{EqBndN1>N2}) in Theorem \ref{ThmBndN1N2}.
Actually, the following theorem states that the code $\mathcal{C}$ in Construction \ref{Cnst} alsways attains the bound (\ref{EqBndN1>N2}) in Theorem \ref{ThmBndN1N2}.
\end{example}

\begin{theorem}\label{ThmOptimalCnst}
The $[n,k]$ LRC $\mathcal{C}$ obtained from Construction \ref{Cnst} has the minimum distance
\begin{equation*}
d = n-k+1-\tilde{\eta},
\end{equation*}
where
$ \tilde{\eta} = \min \{ \left\lceil \frac{(\lambda+1)(k-1)+1}{(\lambda+1)(r-1) + 1} \right\rceil, \left\lceil \frac{\lambda(k-1) + \nu +1}{\lambda(r-1) + 1} \right\rceil\} -1.$
\end{theorem}
\begin{proof}
First, we claim that for any $V \subseteq \Omega$ with $|V| = k+\tilde{\eta}$, there exist subsets $V_1,\dots,V_\mu \subseteq V$ such that the following two conditions are satisfied:
\begin{itemize}
\item[(1)] $|\cup_{l=1}^\mu V_l| \ge k$;
\item[(2)] For $1 \le l \le \mu$, $V_l \subseteq \cup_{i=1}^{\lambda+1} W^{(l)}_i$, and there exists $\omega_l \in V_l$ such that $|(V_l - \{\omega_l\}) \cap W^{(l)}_i| \le r-1$ for all $i \in [\lambda+1]$.
\end{itemize}
The proof of the claim is given in Lemma \ref{LemClm} of Appendix \ref{AppPrfClm}.

From the claim and Lemma \ref{LemAlphaBeta} (ii), we can deduce that, for $1 \le l \le \mu$, the elements in $V_l$ are linearly independent over $\mathbb{F}_q$, and thus the elements in $V_1 \cup V_2 \cup \dots \cup V_\mu$ are linearly independent over $\mathbb{F}_q$.
Then by Proposition \ref{PropFiniteField}, $\mathcal{C}$ can tolerate any $n - (k+ \tilde{\eta})$ erasures.
Consequently, the minimum distance of $\mathcal{C}$ satisfies $d \ge n-k+1-\tilde{\eta}$, and the equality actually holds because of Theorem \ref{ThmBndN1N2}.
\end{proof}

\subsection{Influence of the Regenerating Set Structure}
As we have stated in Example \ref{ExConstruct} and earlier sections, the structure of regenerating sets can influence the value of the function $\Phi(x)$ which in turn influence the value of the minimum distance. In this section, we will check the regenerating set structure of the code $\mathcal{C}$ in Construction \ref{Cnst} to support its attaining the optimal minimum distance, and also make  a comparison with some previously constructed codes.

In Fig. \nolinebreak \ref{FigOmegaStructure} it gives a graphical description of the regenerating sets for $\mathcal{C}$, while each line (or a branch, i.e. $W^{(l)}_i$ ) represents a regenerating set. Consider the collection of regenerating sets $\{W^{(l)}_i \}_{l \in [\mu],i \in [\lambda+1]}$. It has a nontrivial union with respect to any order they are arranged in.

In fact, it is easy to see that for the code $\mathcal{C}$,
$$ \Phi(x) = \Min_{\substack{\mathcal{I} \subseteq \{W^{(l)}_i \}_{l \in [\mu],i \in [\lambda+1]} \\ |\mathcal{I}| = x}}|\cup_{V \in \mathcal{I}} V|.$$
We can count from Fig. \nolinebreak \ref{FigOmegaStructure} that
\begin{align*}
\Min_{\substack{\mathcal{I} \subseteq \{W^{(l)}_i \}_{l,i} \\ |\mathcal{I}| = x}}|\cup_{V \in \mathcal{I}} V| &= \begin{cases} xr+ \left\lceil \frac{x}{\lambda+1} \right\rceil, \text{ if } x \le \nu(\lambda+1)  \\ xr + \nu + \left\lceil \frac{x-\nu(\lambda+1)}{\lambda} \right\rceil, \text{ if } x > \nu(\lambda+1)\end{cases} \\
&= xr+\max\{ \left\lceil \frac{x}{\lambda+1} \right\rceil, \left\lceil \frac{x-\nu}{\lambda} \right\rceil\}.
\end{align*}
Therefore, the $\Phi(x)$ of $\mathcal{C}$ satisfies
\begin{equation*}
\Phi(x) = xr+\max\{ \left\lceil \frac{x}{\lambda+1} \right\rceil, \left\lceil \frac{x-\nu}{\lambda} \right\rceil\},
\end{equation*}
which attains the upper bound defined by $\Psi(x)$ (see Theorem \ref{ThmIntBnd} and Proposition \ref{ProIntN1>N2}). That is, $\mathcal{C}$ achieves the maximum value of $\Phi(x)$ among all the LRCs with $n_1>n_2$, which can be regarded as a support of the code $\mathcal{C}$ attaining the optimal minimum distance.

On the other hand, we will see some previously constructed codes have smaller minimum distance due to their regenerating set structure. The code presented by Silberstein et al. in \cite{silberstein2013optimal} and that proposed by Tamo et al. in \cite{FamilyTamo14} are both of pairwise disjoint regenerating sets. Namely, partition the set $[n]$ into $n_1$ subsets $I_1,I_2,\dots,I_{n_1}$ such that $|I_j| =r+1$ for $1 \le j \le n_1 -1$ and $|I_{n_1}| = r+1 - n_2$, then $I_1,I_2,\dots,I_{n_1}$ form a sequence of regenerating sets that has a nontrivial union.

Clearly, the $\Phi(x)$ satisfies
\begin{equation*}
\Phi(x) \le (r+1)x - n_2 , \;\;\forall 1 \le x \le n_1.
\end{equation*}
Then by Theorem \ref{ThmRegDis}, $\rho = \max\{x : \Phi(x) -x < k \} \ge \left\lceil \frac{k+n_2}{r} \right\rceil -1$, and the minimum distance satisfies
\begin{equation*}
d \le n- k+1 - (\left\lceil \frac{k+n_2}{r} \right\rceil -1).
\end{equation*}
Thus it cannot attain the bound (\ref{EqGoplanBnd}) when $\left\lceil \frac{k+n_2}{r} \right\rceil > \left\lceil \frac{k}{r} \right\rceil$, i.e., $ k \mod r \ge n \mod (r+1) >0$. In fact, the minimum distance  sometimes goes beneath the bound (\ref{EqBndN1>N2}) of Theorem \ref{ThmBndN1N2}, that is, the optimal minimum distance cannot be attained under this kind of regenerating set structure.
Fig. \nolinebreak \ref{ConsComp1} gives a comparison between the minimum distance of $\mathcal{C}$ and that of the codes in \cite{silberstein2013optimal,FamilyTamo14} for $n=25$ and  $r=3$.
\begin{figure}
  \centering
  \includegraphics[width=0.52\textwidth]{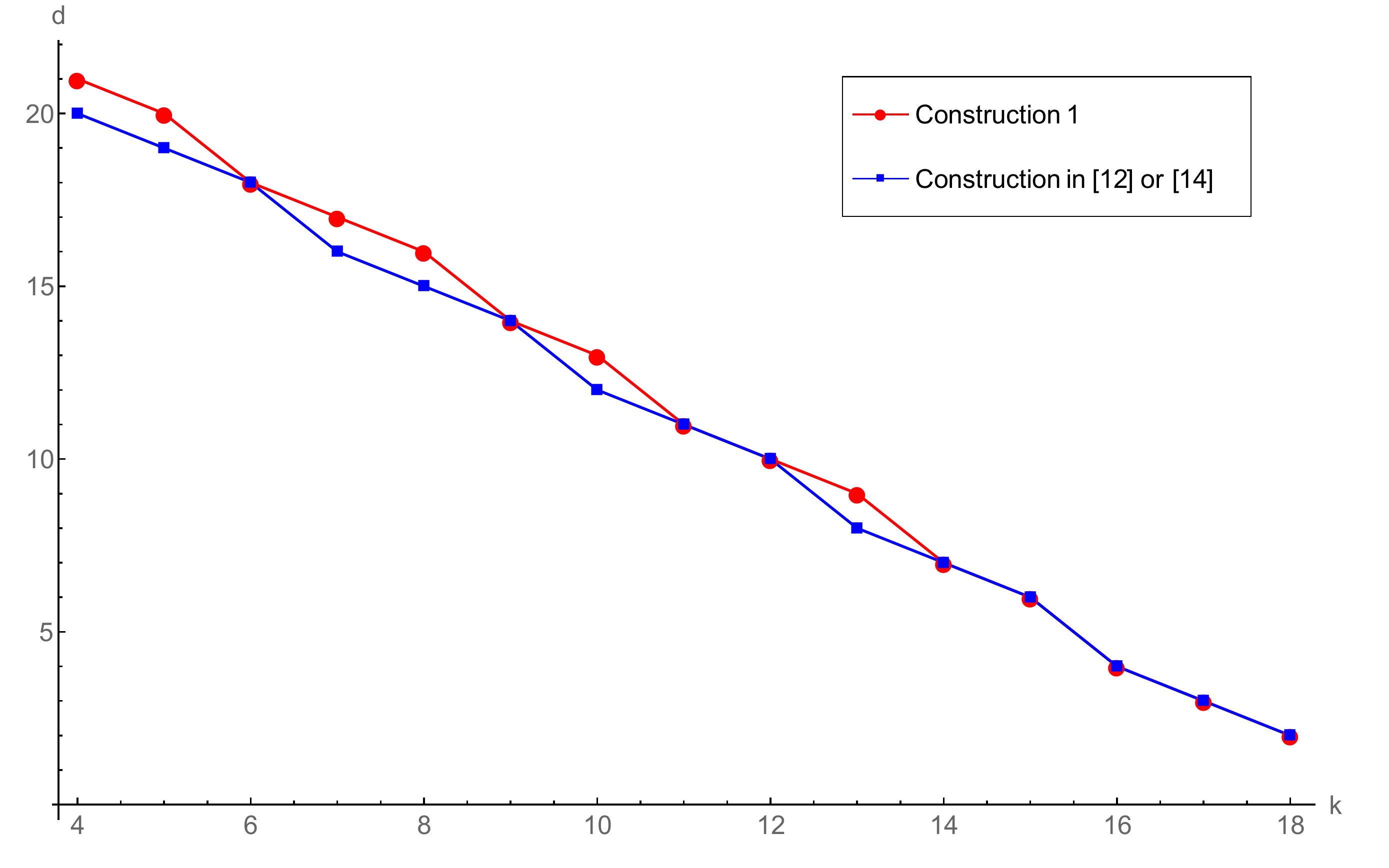}\\
  \caption{A comparison of the two LRCs for $n=25 ,r=3$}\label{ConsComp1}
\end{figure}

\section{Conclusions}\label{SecConclusion}
In this paper we carry out an in-depth study of the two problems:
what is the largest possible minimum distance for an $[n,k]$ LRC?
How to construct an $[n,k]$ LRC with the largest possible minimum distance?
For the first problem, we derive an integer programming based upper bound on the minimum distance for LRCs, and then give an explicit bound by solving the integer programming problem.
The explicit bound applies all LRCs satisfying $n_1 >n_2$ .
For the second problem, we present a construction of linear LRCs that attains the explicit bound for $n_1>n_2$.
Therefore, we have completely solved the two problems under the condition $n_1> n_2$.
However, for $n_1 \le n_2$ the two problems remain unsolved in many cases.

\appendices

\section{Proof of Proposition \ref{ProIntN1>N2}}\label{AppIntN1>N2}
\begin{lemma}
For $1 \le x \le n_1$,
\begin{equation*}
\Psi(x) \ge xr+\max\{ \left\lceil \frac{x}{\lambda+1} \right\rceil, \left\lceil \frac{x-\nu}{\lambda} \right\rceil\}.
\end{equation*}
\end{lemma}
\begin{proof}
Set
\begin{equation*}
\begin{cases}
s = \mu, \\
t_1 = \dots = t_\nu = \lambda+1, t_{\nu+1} = \dots = t_\mu =\lambda, \\
a_1 = \dots = a _\nu = \lambda, a_{\nu+1} = \dots = a_\mu = \lambda-1.
\end{cases}
\end{equation*}
It is clear that  $s$ and $\{t_i,a_i\}_{i \in [s]}$ satisfy (\ref{EqIntCondOut}), and then we have
\begin{align*}
\Psi(x) & \ge \Min_{l,h_1,\dots,h_l} (xr+1 - \sum_{i=1}^{l-1}(a_{h_i} - t_{h_i})) \\
& = \Min_{l,h_1,\dots,h_l} (xr+l),
\end{align*}
where the minimum is subject to (\ref{EqIntCondIn}).
On the other hand, for any integers $l, h_1,\dots,h_l$ satisfying (\ref{EqIntCondIn}),
\begin{equation*}
x \le t_{h_1} + \dots + t_{h_l} \le \begin{cases} (\lambda +1) l, \text{ if } l \le \nu, \\ (\lambda+1)\nu + (l -\nu) \lambda, \text{ if } l > \nu,
\end{cases}
\end{equation*}
which induces $x \le \min \{ \lambda l + l , \lambda l + \nu \}$.
Therefore $l \ge \max\{ \left\lceil \frac{x}{\lambda+1} \right\rceil, \left\lceil \frac{x-\nu}{\lambda} \right\rceil\}$ and then
\begin{equation*}
\Psi(x) \ge \Min_{l,h_1,\dots,h_l} (xr+l) \ge xr + \max\{ \left\lceil \frac{x}{\lambda+1} \right\rceil, \left\lceil \frac{x-\nu}{\lambda} \right\rceil\}.
\end{equation*}
\end{proof}

\begin{lemma}
For $1 \le x \le n_1$,
\begin{equation*}
\Psi(x) \le xr+\max\{ \left\lceil \frac{x}{\lambda+1} \right\rceil, \left\lceil \frac{x-\nu}{\lambda} \right\rceil\}.
\end{equation*}
\end{lemma}
\begin{proof}
We prove the lemma by contradiction.
Assume that for some $1 \le x \le n_1$,
$$\Psi(x) \ge xr+1 +\max\{ \left\lceil \frac{x}{\lambda+1} \right\rceil, \left\lceil \frac{x-\nu}{\lambda} \right\rceil\}.$$
Then there exist integers $s$ and $t_i,a_i$,  $1 \le i \le s$, satisfying the constraints (\ref{EqIntCondOut}) and
$$ \Min_{l,h_1,\dots,h_l} (xr+1 - \sum_{i=1}^{l-1}(a_{h_i} - t_{h_i})) \geq xr+1 +\max\{ \left\lceil \frac{x}{\lambda+1} \right\rceil, \left\lceil \frac{x-\nu}{\lambda} \right\rceil\},$$
where the minimum is subject to the constraint (\ref{EqIntCondIn}).
Therefore for all integers $l$ and $h_1,\dots,h_l \in [s]$ satisfying the constraint (\ref{EqIntCondIn}), it has
\begin{equation}\label{EqIntCon}
\sum_{i=1}^{l-1}(a_{h_i} - t_{h_i}) \le - \max\{ \left\lceil \frac{x}{\lambda+1} \right\rceil,\left\lceil \frac{x-\nu}{\lambda} \right\rceil\}\;.
\end{equation}

Consider the following two cases.

{\bf Case 1.} $ 1\le x \le (\lambda+1) \nu$.
Then $\max\{ \left\lceil \frac{x}{\lambda+1} \right\rceil,\left\lceil \frac{x-\nu}{\lambda} \right\rceil\} = \left\lceil \frac{x}{\lambda+1} \right\rceil.$
For $1 \le i \le s$, define
$$b_i = (\lambda +1)a_i - \lambda t_i.$$
Then without loss of generality, we can assume that $b_1 \ge b_2 \ge \dots \ge b_s$.
Let $h$ be the smallest integer such that $1 \le h \le s$ and $ \sum_{i=1}^h (a_i - t_i) \le -\left\lceil \frac{x}{\lambda+1} \right\rceil.$
Note that $h$ exists because $\sum_{i=1}^s (a_i - t_i) = n_2 - n_1 =  - \mu  \le -\left\lceil \frac{x}{\lambda+1} \right\rceil$.
Next we consider the value of $t_1 + \dots + t_h$.

If $t_1 + \dots + t_h \ge x$, there exists a positive integer $h^\prime \le h$ such that $\sum_{j=1}^{h^\prime -1}t_j < x \le \sum_{j=1}^{h^\prime }t_j$.
Then $h'-1<h$ and by (\ref{EqIntCon}),
$$\sum_{i=1}^{h^\prime-1}(a_i - t_i) \le -  \left\lceil \frac{x}{\lambda+1} \right\rceil,$$
which contradicts to the minimality of $h$.

If $t_1 + \dots + t_h < x$, we compute $\sum_{i=1}^s b_i$ in two different ways.
On the one hand,
\begin{align}\label{EqSumBi}
\sum_{i=1}^s b_i & = \sum_{i=1}^s ((\lambda +1)a_i - \lambda t_i) \nonumber \\
& = (\lambda +1) n_2 - \lambda n_1 \nonumber \\
& = \nu - \mu.
\end{align}
On the other hand, we claim that
\begin{enumerate}
\item[(i)] $\sum_{i=1}^h b_i \le -1$ and $b_i \le -1$ for $h +1 \le i \le s$,
\item[(ii)] $h -s \le \nu - \mu$,
\end{enumerate}
and then
\begin{align*}
\sum_{i=1}^s b_i & = (\sum_{i=1}^h b_i) + (\sum_{i=h+1}^s b_i) \\
& \le -1 + (-1) \times (s-h) \\
& = -1+h-s \\
& \le \nu - \mu -1,
\end{align*}
which contradicts to (\ref{EqSumBi}).

In fact, the claim (i) holds because
\begin{align*}
\sum_{i=1}^h b_i & = \sum_{i=1}^h ((\lambda +1)a_i - \lambda t_i) \\
& = (\lambda +1)\sum_{i=1}^h (a_i - t_i ) + \sum_{i=1}^h t_i \\
& \substack{(a) \\ \le} -(\lambda +1)\left\lceil \frac{x}{\lambda+1} \right\rceil + x-1 \\
& \le -1,
\end{align*}
where (a) follows from $\sum_{i=1}^h (a_i - t_i ) \le -\left\lceil \frac{x}{\lambda+1} \right\rceil$ and $\sum_{i=1}^h t_i < x$.
Then $b_j \le \frac{1}{h} \sum_{i=1}^h b_i <0$ for $h+1 \le j \le s$.
To show the claim (ii), observe that $a_i \ge t_i -1$ and $\sum_{i=1}^{h-1} (a_i - t_i) \ge -\left\lceil \frac{x}{\lambda+1} \right\rceil +1$ from the minimality of $h$.
Then we have
\begin{align*}
- \mu & = n_2 - n_1 = \sum_{i=1}^s (a_i - t_i) \\
& = \sum_{i=1}^{h-1} (a_i - t_i) + \sum_{i=h}^s (a_i - t_i) \\
& \ge -\left\lceil \frac{x}{\lambda+1} \right\rceil + 1 + (-1) \times (s-h+1).
\end{align*}
Because $x \le (\lambda +1 )\nu$, it holds
$$ - \mu \ge - \nu +1 + (-1) \times (s-h+1) = - \nu + h -s,$$
and the claim (ii) follows directly.

{\bf Case 2.} $(\lambda+1)\nu+1 \leq x \le n_1$.
Then $\max\{ \left\lceil \frac{x}{\lambda+1} \right\rceil,\left\lceil \frac{x-\nu}{\lambda} \right\rceil\} = \left\lceil \frac{x-\nu}{\lambda} \right\rceil.$
Similar to Case 1, define
$$c_i = \lambda a_i - (\lambda -1)t_i ,\;\;\;\forall\; 1\le i \le s$$
and assume $c_1 \ge c_2 \ge \dots \ge c_s$.
Let $g \in [s]$ be the smallest positive integer such that $ \sum_{i=1}^{g}(a_i - t_i) \le -\left\lceil \frac{x-\nu}{\lambda} \right\rceil.$
Note that $g$ exists because $\sum_{i=1}^s (a_i - t_i) = n_2 - n_1 =  - \mu  \le -\left\lceil \frac{x-\nu}{\lambda} \right\rceil$.
Next we consider the value of $t_1 + \dots + t_g$.

Similar to Case 1, $t_1 + \dots + t_g \ge x$ contradicts to the minimality of $g$.
Then it follows $t_1 + \dots + t_g < x$.
We compute the value of $\sum_{i=1}^{s} c_i$ in two different ways.
On the one hand,
\begin{align}\label{EqSumCi}
\sum_{i=1}^s c_i &= \lambda \sum_{i=1}^s a_i - (\lambda-1) \sum_{i=1}^s t_i \nonumber \\
&= \lambda n_2 - (\lambda-1) n_1 \nonumber \\
& = \nu.
\end{align}
On the other hand, we claim that
\begin{enumerate}
\item[(i)] $\sum_{i=1}^{g} c_i \le \nu-1$,
\item[(ii)] $c_i \le 0$ for  $g+1 \le i \le s$.
\end{enumerate}
Then
\begin{align*}
\sum_{i=1}^s c_i & = \sum_{i=1}^g c_i + \sum_{i=g+1}^s c_i \\
& \le \sum_{i=1}^g c_i \le \nu -1,
\end{align*}
which contradicts to (\ref{EqSumCi}).

Note that $\sum_{i=1}^g (a_i - t _i) \le -\left\lceil \frac{x-\nu}{\lambda} \right\rceil$ and $\sum_{i=1}^{g} t_i < x$, then the claim (i) follows from
\begin{align*}
\sum_{i=1}^{g} c_i& = \lambda \sum_{i=1}^{g} (a_i - t_i) + \sum_{i=1}^{g} t_i \\
& \le - \lambda \left\lceil \frac{x-\nu}{\lambda} \right\rceil + x-1 \\
& \le \nu-1.
\end{align*}
To show the claim (ii), observe that $c_j \le \frac{1}{g} \sum_{i=1}^g c_i \le \frac{\nu-1}{g}$ for $g+1 \le j \le s$, and $g \ge - \sum_{i=1}^g (a_i - t _i)  \ge \left\lceil \frac{x-\nu}{\lambda} \right\rceil > \nu$ where the first inequality is from $a_i\geq t_i-1$ for $1\leq i\leq s$ and the last inequality is from $x\geq (\lambda+1)\nu+1$.
Then it has
$ c_j < \frac{\nu}{g} <1$ for $g+1 \le j \le s$ and the claim (ii) then follows.
\end{proof}

\section{Proof of Lemma \ref{LemAlphaBeta}}\label{AppPrfLemAlpBet}
(i) Because $X=(\vec{x}_0,\vec{x}_1,\dots,\vec{x}_r)$ generates an $[r+1,r]$ MDS code, there exist nonzero elements $e_0,e_1,\dots,e_r \in \mathbb{F}_q$ such that $e_0 \vec{x}_0 + e_1 \vec{x}_1 + \dots + e_r \vec{x}_r =0$.
Moreover, since $\vec{c}=(1,c_1,\dots,c_r)$ is a codeword of the MDS code, it has $e_0 + e_1c_1 + \dots + e_rc_r = 0$.
Therefore $e_0 \vec{\alpha}_0 + e_1 \vec{\alpha}_{i,1} + \dots + e_r \vec{\alpha}_{i,r} =0$ for $1 \le i \le \lambda+1$.
Thus (i) follows directly.

(ii) We prove the statement by contradiction.
Assume that the vectors in $F$ are linearly dependent, i.e. there exists $e_{\vec{\alpha}}\in\mathbb{F}_q$ for each $\vec{\alpha} \in F$ such that $\sum_{\alpha \in F} e_{\vec{\alpha}} \vec{\alpha} = 0$, where $\{e_{\vec{\alpha}}\}_{\vec{\alpha}\in F}$ are not all zeros.
In fact, at least two out of $\{e_{\vec{\alpha}}\}_{\vec{\alpha}\in F}$ are nonzero because the vectors in $F$ are not zero vectors.
We consider the following two cases.

{\bf Case 1.} $|(F - \{\vec{\alpha}_0\}) \cap \mathcal{A}_i| \le r-1$ for $1 \le i \le \lambda+1$.
Because at least two out of $\{e_{\vec{\alpha}}\}_{\vec{\alpha}\in F}$ are nonzero, there exists $i_0 \in [\lambda+1]$ such that the coefficients $\{e_{\vec{\alpha}}\}_{\vec{\alpha} \in \mathcal{A}_{i_0} \setminus \{\vec{\alpha}_0\}}$ are not all zero.
Then without loss of generality, assume $(F - \{\vec{\alpha}_0\}) \cap \mathcal{A}_{i_0} = \{\vec{\alpha} _{i_0,1},\vec{\alpha} _{i_0,1}\dots,\vec{\alpha} _{i_0,h}\}$, where $h \le r-1$.
Consider the restriction of the linear combination $\sum_{\vec{\alpha}  \in F} e_{\vec{\alpha}} \vec{\alpha} $ to its  $i_0$th thick row, (i.e., the $((i_0-1)r+1)$-th row to the $i_0 r$-th row,) we have $\sum_{j=1}^{h} e_{\vec{\alpha}_{i_0,j}} \vec{x}_j = e \vec{x}_0$ for some $e \in \mathbb{F}_q$.
It follows that $\vec{x}_0,\vec{x}_1,\dots,\vec{x}_h$ are $\mathbb{F}_q$-linearly dependent, where $h\leq r-1$, which contradicts the fact that $(\vec{x}_0,\dots,\vec{x}_r)$ generates an $[r+1,r]$ MDS code.

{\bf Case 2.} For some $(i_0,j_0) \in [\lambda+1] \times [r]$, $|(F - \{\vec{\alpha}_{i_0,j_0}\}) \cap \mathcal{A}_i| \le r-1$ for $1 \le i \le \lambda+1$.
Without loss of generality, assume $i_0 = j_0= 1$, i.e., $|(F - \{\vec{\alpha}_{1,1}\}) \cap \mathcal{A}_i| \le r-1$ for $1 \le i \le \lambda+1$.
If there exists $l$, $2 \le l \le \lambda+1$, such that $\{e_{\vec{\alpha}}\}_{\vec{\alpha} \in \mathcal{A}_{l}\setminus \{\vec{\alpha}_0\}}$ are not all zero, then similar to Case 1, restricting the linear combination $\sum_{\vec{\alpha} \in F} e_{\vec{\alpha}} \vec{\alpha}$ to its $l$th thick row will lead a contradiction.
Therefore we have $e_{\vec{\alpha}}=0$ for all $\vec{\alpha} \in \cup_{i=2}^{\lambda+1} \mathcal{A}_{i} - \{\vec{\alpha}_0\}$.
Thus it suffice to check the vectors in $F\cap \mathcal{A}_1$.
Similarly, a contradiction arises when restricting $\sum_{\vec{\alpha} \in F} e_{\vec{\alpha}} \vec{\alpha}$ to the first thick row.

\section{Proof of The Claim}\label{AppPrfClm}
\begin{lemma}\label{LemClm}
For any $V \subseteq \Omega$ with $|V| = k+\tilde{\eta}$, there exist subsets $V_1,\dots,V_\mu \subseteq V$ such that the following two conditions are satisfied:
\begin{itemize}
\item[(1)] $|\cup_{l=1}^\mu V_l| \ge k$;
\item[(2)] For $1 \le l \le \mu$, $V_l \subseteq \cup_{i=1}^{\lambda+1} W^{(l)}_i$, and there exists $\omega_l \in V_l$ such that $|(V_l - \{\omega_l\}) \cap W^{(l)}_i| \le r-1$ for all $i \in [\lambda+1]$.
\end{itemize}
\end{lemma}
\begin{proof}
Denote $U_l = V \cap (\cup_{i=1}^{\lambda+1} W^{(l)}_i)$ for $1 \le l \le \mu$.
Then the proof is completed by two steps.
First, we show that for all nonempty set $U_l$, $1 \le l \le \mu$, there exists a subset $V_l \subseteq U_l$ satisfying
\begin{itemize}
\item $|V_l| \ge |U_l| - \left\lfloor \frac{|U_l|-1}{r} \right\rfloor$; and
\item There exists $\omega_l \in V_l$ such that $|(V_l - \{\omega_l\}) \cap W^{(l)}_i| \le r-1$ for all $i \in [\lambda+1]$.
\end{itemize}
Second, by setting $V_l = \emptyset$ for all $l \in [\mu]$ with $|U_l| = 0$, we prove that $|V_1 \cup V_2 \cup \dots \cup V_\mu| \ge k$. The details are given below.

{\bf Step 1.} Suppose $U_l$ is nonempty.
Consider the following two cases.

(a) $\omega^{(l)}_0 \in U_l$.
Then there are at most $\left\lfloor \frac{| U_l | -1}{r} \right\rfloor$ sets out of $W^{(l)}_1,W^{(l)}_2,\dots,W^{(l)}_{\lambda+1}$ which are contained in $U_l$, say, $W^{(l)}_1,...,W^{(l)}_h\subseteq U_l$, where $h \le \left\lfloor \frac{| U_l | -1}{r} \right\rfloor$.
Define $V_l$ by deleting $\omega^{(l)}_{1,1},\omega^{(l)}_{2,1},\dots,\omega^{(l)}_{h,1}$ from $U_l$, then we have $|(V_l - \{\omega^{(l)}_0\}) \cap W^{(l)}_i| \le r-1$ for all $i \in [\lambda+1]$ and $|V_l| \ge |U_l| - \left\lfloor \frac{|U_l|-1}{r} \right\rfloor$.

(b) $\omega^{(l)}_0 \notin U_l$.
Similarly, there are at most $\left\lfloor \frac{| U_l |}{r} \right\rfloor$ sets out of $W^{(l)}_1,W^{(l)}_2,\dots,W^{(l)}_{\lambda+1}$ which are contained in $U_l \cup \{\omega^{(l)}_0\}$, say, $W^{(l)}_1,...,W^{(l)}_{h'}\subseteq U_l\cup \{\omega^{(l)}_0\}$, where $h^\prime \le \left\lfloor \frac{| U_l |}{r} \right\rfloor$.
Define $V_l$ by deleting $\omega^{(l)}_{2,1},\omega^{(l)}_{3,1},\dots,\omega^{(l)}_{h^\prime,1}$ from $U_l$, then we have $|(V_l - \{\omega^{(l)}_{1,1}\}) \cap W^{(l)}_i| \le r-1$ for all $i \in [\lambda+1]$, and
$|V_l|  \ge |U_l| - (\left\lfloor \frac{| U_l |}{r} \right\rfloor -1) \ge |U_l| - \left\lfloor \frac{|U_l|-1}{r} \right\rfloor$.

{\bf Step 2.} Observe that
\begin{align*}
|\cup_{l=1}^{\mu} V_l| & = \sum_{l \in [\mu], U_l \neq \emptyset} |V_l| \\
& \ge \sum_{l \in [\mu], U_l \neq \emptyset} (|U_l| - \left\lfloor \frac{|U_l|-1}{r} \right\rfloor) \\
& = k+ \tilde{\eta} - \sum_{l \in [\mu], U_l \neq \emptyset}\left\lfloor \frac{|U_l|-1}{r} \right\rfloor. \\
& \ge k+\tilde{\eta} - \left\lfloor\sum_{l \in [\mu], U_l \neq \emptyset} \frac{|U_l|-1}{r} \right\rfloor \\
& = k+ \tilde{\eta} - \left\lfloor \frac{k+\tilde{\eta}-\epsilon}{r} \right\rfloor,
\end{align*}
where $\epsilon = |\{l \in [l] : U_l \neq \emptyset\}|$.
Then it suffices to show $\left\lfloor \frac{k+\tilde{\eta}-\epsilon}{r} \right\rfloor \le \tilde{\eta}.$

Denote $\epsilon_1 = |\{l :1 \le l \le \nu, U_l \neq \emptyset\}|$ and $\epsilon_2 = |\{l :\nu+1 \le l \le \mu, U_l \neq \emptyset\}|$, then $\epsilon = \epsilon_1 + \epsilon_2$.
Because $|U_1|+|U_2|+\dots+|U_\mu|=k+\tilde{\eta}$ and
\begin{equation*}
|U_l| \le \begin{cases} |\cup_{i=1}^{\lambda+1} W^{(l)}_i| = (\lambda+1)r+1, \text{ for } 1 \le l \le \nu \\ |\cup_{i=1}^{\lambda} W^{(l)}_i| = \lambda r+1, \text{ for } \nu+1 \le l \le \mu,\end{cases}
\end{equation*}
we have
\begin{equation*}\label{EqEpsilon}
\begin{cases}
0 \le \epsilon_1 \le \nu; \\
0 \le \epsilon_2 \le \mu-\nu; \\
k+\tilde{\eta} \le \epsilon_1 ((\lambda+1)r+1) + \epsilon_2 (\lambda r+1).
\end{cases}
\end{equation*}
Since $\epsilon_1 ((\lambda+1)r+1) + \epsilon_2 (\lambda r+1)\leq ((\lambda+1)r+1)(\epsilon_1+\epsilon_2)$ and also $\epsilon_1 ((\lambda+1)r+1) + \epsilon_2 (\lambda r+1)= (\lambda r+1)(\epsilon_1+\epsilon_2)+\epsilon_1r\le(\lambda r+1)(\epsilon_1+\epsilon_2)+\nu r$,
it follows that $\epsilon \ge \max\{ \frac{k+\tilde{\eta}}{(\lambda+1)r+1} , \frac{k+\tilde{\eta}- r\nu}{\lambda r+1}\}$.
Thus
\begin{align*}
\left\lfloor \frac{k+\tilde{\eta}-\epsilon}{r} \right\rfloor & \le \left\lfloor \frac{1}{r}(k+\tilde{\eta}-\max\{ \frac{k+\tilde{\eta}}{(\lambda+1)r+1} ,  \frac{k+\tilde{\eta}- r\nu}{\lambda r+1} \}) \right\rfloor \\
& = \left\lfloor \frac{1}{r} \min\{ \frac{(k+\tilde{\eta})(\lambda+1)r}{(\lambda+1)r+1} ,  \frac{(k+\tilde{\eta})\lambda r + r\nu}{\lambda r+1} \} \right\rfloor \\
& = \min\{\left\lfloor \frac{(k+\tilde{\eta})(\lambda+1)}{(\lambda+1)r+1} \right\rfloor, \left\lfloor \frac{(k+\tilde{\eta})\lambda + \nu}{\lambda r+1} \right\rfloor\}.
\end{align*}
Note that $\tilde{\eta} = \min \{ \left\lceil \frac{(\lambda+1)(k-1)+1}{(\lambda+1)(r-1) + 1} \right\rceil, \left\lceil \frac{\lambda(k-1) + \nu +1}{\lambda(r-1) + 1} \right\rceil\} -1$.
Then if $\tilde{\eta} = \left\lceil \frac{(\lambda+1)(k-1)+1}{(\lambda+1)(r-1) + 1} \right\rceil-1$, it has
\begin{align*}
\frac{(k+\tilde{\eta})(\lambda+1)}{(\lambda+1)r+1} - (\tilde{\eta}+1) & = \frac{(\lambda+1)(k-1)-((\lambda+1)(r-1)+1)(\tilde{\eta}+1)}{(\lambda+1)r+1} \\
& = \frac{(\lambda+1)(r-1)+1}{(\lambda+1)r+1} \times (\frac{(\lambda+1)(k-1)}{(\lambda+1)(r-1)+1} - (\tilde{\eta}+1)) \\
& < 0,
\end{align*}
and therefore $\left\lfloor \frac{(k+\tilde{\eta})(\lambda+1)}{(\lambda+1)r+1} \right\rfloor \le \tilde{\eta}$.
Similarly, if $\tilde{\eta} = \left\lceil \frac{\lambda(k-1) + \nu +1}{\lambda(r-1) + 1} \right\rceil-1$, it can be proved that $\left\lfloor \frac{(k+\tilde{\eta})\lambda + \nu}{\lambda r+1} \right\rfloor \le \tilde{\eta}$.
Thus we conclude that $\left\lfloor \frac{k+\tilde{\eta}-\epsilon}{r} \right\rfloor \le \tilde{\eta}$.
\end{proof}

\end{document}